\documentclass[a4paper,11pt]{article}
\usepackage{jheppub} 
\usepackage{lineno}

\usepackage{mathrsfs}
\usepackage{graphicx}
\usepackage{amsmath}
\usepackage{fancybox,ascmac}
\usepackage{amsthm}
\usepackage{mathrsfs}
\usepackage{braket}
\usepackage{mathtools}
\usepackage{tikz}
\usetikzlibrary{intersections, calc, arrows.meta}
\usepackage{physics}
\usepackage{hyperref}

\makeatletter
\@addtoreset{equation}{section}
\makeatother

\makeatletter 
\renewcommand{\figurename}{Fig}
\renewcommand{\thefigure}{\thesection.\arabic{figure}}
\@addtoreset{figure}{section}
\makeatother 

\makeatletter 
\renewcommand{\tablename}{Table}
\renewcommand{\thetable}{\thesection.\arabic{table}}
\@addtoreset{table}{section}
\makeatother 

\let\Re\relax
\DeclareMathOperator{\Re}{Re}
\let\Im\relax
\DeclareMathOperator{\Im}{Im}

\DeclareMathOperator{\SL}{SL}

\newtheorem{thm}{Theorem}[section]

\newtheorem{prop}[thm]{Proposition}

\newtheorem{lem}[thm]{Lemma}

\title{\boldmath On the classification of duality defects in 
$c=2$ compact boson CFTs with a discrete group
orbifold}

\author{Yuma Furuta}
\affiliation{Research Institute for Mathematical Sciences, Kyoto University,\\
Kyoto, Japan}

\emailAdd{yfuruta@kurims.kyoto-u.ac.jp}

\abstract{We propose a novel approach to exploring duality defects in the $c=2$ compact boson conformal field theory (CFT). This study is motivated by the desire to classify categorical symmetries, particularly duality defects, in CFTs. While the $c=1$ case has been extensively studied, and the types of realizable duality defects are largely understood, the situation becomes significantly more complex for $c=2$. The simplicity of the $c=1$ case arises from the fact that its theory is essentially determined by the radius of compactification. In contrast, the $c=2$ case involves more parameters, leading to a more intricate action of T-duality.
As a result, directly solving the condition for a theory to be self-dual under orbifolding becomes highly challenging. To address this, we categorize duality defects into four types and demonstrate that the condition for a toroidal branch theory to be self-dual under an orbifold induced by an automorphism generated by shift symmetry can be reformulated as quadratic equations. We also found that for ``almost all" theories we can enumerate all solutions for such equations. Moreover,
this reformulation enables the simultaneous exploration of multiple duality defects and provides evidence for the existence of duality defects under specific parameter families for the theory, such as $(\tau, \rho) = (it, \frac{1}{2}+it)$ where $t \in \mathbb{Q}$.}

\begin{document}
\maketitle
\flushbottom

\section{Introduction}
\label{sec:intro}
Symmetry is a highly useful concept for studying physics, including quantum field theory, and has been extensively explored by many physicists and mathematicians. For instance, symmetry has been employed to analyze phase transitions triggered by symmetry breaking and to classify phases, as well as to categorize non-perturbative effects (primarily instantons). Over the past few decades, the concept has been significantly broadened in scope \cite{gaiotto_generalized_2015}.
In this context, our focus is on non-invertible symmetry, a type of symmetry that does not necessarily have inverses, unlike conventional group-theoretic symmetries. Such symmetries have been discussed across various dimensions in recent years, with one notable example being the order-disorder operator in the two-dimensional conformal field theory of the Ising model. At first glance, this appears to correspond to a $\mathbb{Z}_2$ symmetry that interchanges the ordered and disordered phases. However, this transformation lacks a true inverse.

In fact, while the ordered phase is characterized by the alignment of spins in a single direction, this information is lost when mapped to the disordered phase. Consequently, applying this operator again does not return the system to its original state. Instead, it produces a superposition of states where the spins are aligned either in the $+\frac{1}{2}$ or $-\frac{1}{2}$ direction. A concrete computation reveals that if we denote this order-disorder operator by $\sigma$, then
\begin{equation}\label{eq:1.1}
    \sigma\times\sigma=1+\epsilon
\end{equation}
where $\epsilon$ is an operator that flips the spin directions.
This type of algebraic structure, known as a fusion algebra, has been known in the context of CFTs for a long time. However, it has gained significant attention in recent decades, being applied and studied in various quantum field theories.

Moreover, such generalized symmetries have increasingly been interpreted as topological defects realized in spacetime. In the context of traditional global symmetries, symmetry has been formalized as an operator acting on the Hilbert space that commutes with the Hamiltonian. According to Noether's theorem, the conserved charges corresponding to such symmetries are defined on ``surfaces" that are one dimension lower than the spacetime. Reinterpreting this, one can regard the surface-associated charged operator as acting on a local operator. 
In this sense, symmetry can be defined as specifying how a surface (or defect) acts on a local operator.  
Furthermore, due to their commutativity with the Hamiltonian, such defects can be continuously deformed freely, as long as they do not cross insertion points or singularities. In physical terms, this means they are topological. As a result, the study of global symmetries has transitioned into the study of topological defects. Additionally, fusion rules such as Equation \eqref{eq:1.1} are defined in terms of the limiting behavior when two defects are brought close together and allowed to collide.

Among the topological defects in quantum field theory, we focus on the defects that appear in conformal field theory (CFT). CFTs possess the special symmetry of conformal invariance, and their scale invariance makes them highly useful for studying gapless theories that are fixed points of renormalization group flows. In particular, two-dimensional CFTs have the infinite-dimensional Lie algebra known as the Virasoro algebra as their symmetry, which has deep mathematical properties and a close relationship with physics through vertex operator algebras \cite{frenkel_vertex_1988}.  

Notably, the topological defects of CFTs have been extensively studied since the works of \cite{verlinde_fusion_1988} and \cite{fuchs_boundaries_2003,fuchs_tft_2002,fuchs_tft_2004,fuchs_tft_2004-2,fuchs_tft_2005,fjelstad_tft_nodate}. For example, in \cite{frohlich_defect_2010,frohlich_duality_2007,fischer_classification_2013}, algebraic structures (particularly Frobenius algebras) formed by CFT topological defects are used to construct correlation functions and to attempt a reformulation of rational CFTs. In such cases, the algebra formed by the defects encodes how the left and right chiral algebras are glued together to construct the full CFT. This connects to earlier work by Kreuzer and Schellekens \cite{kreuzer_simple_1994}, who mathematically studied simple current orbifolds. 
Consequently, questions regarding the algebraic information of a given CFT—such as what topological defects exist and what fusion rules they satisfy—are not only of mathematical interest, involving the classification of module categories, but also hold significant physical meaning.

As a simple example, let us consider a massless compact boson with a central charge of 1. Such CFTs have been completely classified by Ginsparg \cite{ginsparg_curiosities_1988} and consist of two types of branches and three exceptional isolated points. These CFTs include several physically significant cases, such as the SU(2)$_1$ WZW model and the tensor product of two Ising models. Exploring the generalized symmetries of these CFTs is thought to aid in understanding the RG behavior of various QFTs. Indeed, as mentioned earlier, these symmetries are topological and hence invariant under RG flows. For example, studies by \cite{thorngren_fusion_2024-2,thorngren_fusion_2024-3} have explored non-invertible symmetries in such CFTs and even computed the RG flow induced by certain deformations. These computations are significant not only for particle physics but also for theories of topological materials. However, as the central charge exceeds 2, the problem becomes significantly more complex and challenging to solve. 
In fact, as explained in Section \ref{sec:2}, this problem involves solving the condition under which a theory is self-dual under a specific orbifold construction. However, as described at the beginning of Section \ref{sec:3.1}, the degrees of freedom grow quadratically with the central charge when it is 2 or greater, making the problem considerably harder to address. Nevertheless, simplifications are possible. For instance, in \cite{nagoya_non-invertible_2023}, computations are restricted to cases where the B-field is zero, while in \cite{damia_exploring_2024}, duality defects are determined only at specific points, such as bicritical and quadricritical points.

This paper aims to generalize and resolve such scenarios for the toroidal branch of theories with $c=2$. More specifically, we propose a set of quadratic equations that detect duality defects based on the condition of self-duality under orbifolding by discrete subgroups of a certain U(1) symmetry. That is, we claim that each integer solution of these quadratic equations corresponds to a duality defect. This allows us to reduce the problem, which involves many degrees of freedom and matrix equations, to a simpler problem involving equations with two variables, reframing it as a question of finding integer solutions.  

Moreover, we achieve a significant result: in most cases, all integer solutions to the quadratic equations can be enumerated. The details of this result are given in Proposition \ref{prop:3.2}, which implies that almost all duality defects under consideration can be exhaustively listed. While not explicitly performed in this work, fusion rules can also be computed using the resulting integer solutions and the corresponding $\SL(2,\mathbb{Z})$ matrices, as shown, for example, in \cite{damia_exploring_2024}.

The structure of this paper is as follows. In Section \ref{sec:2}, we briefly review the self-duality of $c=1$ compact boson CFTs under shift symmetry orbifolds, providing a few simple examples to observe some duality defects. In Section \ref{sec:3}, we extend these results to $c=2$. Section \ref{sec:3.1} formulates the problem to be solved, while Section \ref{sec:3.2} shows that finding duality defects is equivalent to deriving integer solutions to a certain set of quadratic equations, simplifying the problem into a mathematical one. Section \ref{sec:3.3} demonstrates that, in most cases, these solutions can be enumerated, showing that for almost all points in the fundamental domain of the complex and Kähler moduli of the torus, there are precisely two solutions. We verify through simple examples that these are indeed all solutions.  
In Section \ref{sec:4}, we apply these results to demonstrate the existence of duality defects at a broader range of points. These points belong to a set called multicritical points, which are intersections of multiple conformal manifolds. Identifying exactly marginal deformations from such points could also hint at the existence of non-invertible symmetries in the branches along those directions. However, such calculations are left for future work.

\section{The Orbifolds of the Compact Boson CFTs}\label{sec:2}
In this chapter, we will first translate our logical flow into a simpler setup. Specifically, we will discuss what kind of duality defect appears under the orbifold of a discrete group in the case where the central charge is 1 instead of 2, such that it becomes self-dual. This discussion has already been considered by many researchers, and a more detailed understanding can be obtained by reading, for example, \cite{ginsparg_applied_1991}.

\subsection{Review of $c=1$ compact boson CFT}\label{sec:2.1}
First, we will begin with a brief review of the $c=1$ CFT. This is a bosonic theory with a target space of $S^1$ with radius $R$
\footnote{Specifically, there are two main types of such CFTs. The first is a boson compactified on $S^1$, and the second is the orbifold of these by a charge conjugation symmetry. In this section, we refer to the circle branch, i.e., the theory compactified on $S^1$, as the $c=1$ theory. Of course, the same discussion can be applied to the orbifold branch as well.}, described by the following action.
\begin{equation}
    S=\frac{R^2}{4\pi}\int \dd z\dd \bar{z} \partial X(z,\bar{z})\bar{\partial}X(z,\bar{z}).
\end{equation}
This theory has conformal symmetry and also possesses a global symmetry $U(1) \times U(1)$. This holds for compact boson CFTs with general central charge $c$, so when we refer to a primary field below, it will be with respect to this extended $U(1)^c \times U(1)^c$ symmetry. It is known that there is a correspondence between operators and states in CFTs, and instead of determining the Hilbert space of the theory, it is sufficient to list the primary vertex operators that appear in the theory\cite{kac_vertex_1998}. In the case of $c=1$, any vertex primary field has a discretized spectrum.
\begin{equation}
    (p_L,p_R)=\left( \frac{n}{R}+wR,\frac{n}{R}-wR \right),\quad
    n,w\in\mathbb{Z}.
\end{equation}
That is, any $V$ is labeled by two integers. For the purpose of the following discussion, we will rewrite the labels of the chiral and anti-chiral momentum parts, $(p_L, p_R)$, and introduce new labels $(\theta, \phi)$:
\begin{equation}
    \begin{split}
        \theta&=\frac{1}{R}(X_L+X_R)\\
        \phi &=R(X_L-X_R).
    \end{split}
\end{equation}
where $X_L$ and $X_R$ are left moving sector of $X$ and  respectively right moving sector.
Expressing it in this form allows us to represent the two $U(1)$ symmetry transformations in this CFT as shifts of the two periodic parameters $(\theta, \phi)$. Therefore, we denote the global symmetry of this CFT as $U(1)_{\theta} \times U(1)_{\phi}$.

Next, let us discuss the orbifold. The discrete groups we consider this time are all finite subgroups $U(1)_{\theta} \times U(1)_{\phi}$ of $G_{N,W} := \mathbb{Z}_N \times \mathbb{Z}_W$. We need to clarify how the generators of this group act. In principle, any choice of the subgroup $G_{N,W}$ can be applied to our scenario, but for simplicity, we will consider the diagonal subgroup, where the generators of $\mathbb{Z}_N$ and $\mathbb{Z}_W$ act only on $U(1)_{\theta}$ and $U(1)_{\varphi}$, respectively. That is, we consider a group action such as
\begin{equation}\label{eq:2.4}
    G_{N,W}\ni (a,b):(\theta,\phi)\mapsto(\theta+\frac{2\pi a}{N},\phi+\frac{2\pi b}{W}).
\end{equation}
Below, we will show that the orbifold by this diagonal subgroup $G_{N,W}$ changes the radius of the target space $S^1$ from $R$ to $\frac{W}{N}R$.

Let us first review the operation of orbifolding. This can be applied to any general quantum field theory (QFT), and when a symmetry is acting on the QFT—whether a conventional group-theoretic symmetry or a more general one—a new theory can emerge. Naturally, starting from a conformal field theory (CFT), orbifolding leads to another CFT. A well-known example is the $\mathbb{Z}_2$ orbifold for the $c=1$ case. This is an orbifold under charge conjugation $X \mapsto -X$, which is a global symmetry distinct from the $U(1)$ symmetry present in the theory. Geometrically, this operation corresponds to identifying the upper and lower parts of the circle on which the boson is compactified.
In CFT, such an operation is defined as follows. Start with the partition function of the CFT:
\begin{equation}
    \mathcal{Z}(\tau,\bar{\tau})=\tr _{\mathcal{H}}(q^{L_0-\frac{c}{24}}\bar{q}^{\bar{L}_0-\frac{c}{24}}),\quad q=e^{2\pi i\tau},
\end{equation}
where $\tau$ is the torus parameter of the worldsheet.
Here, a finite group $H$ that acts on the Hilbert space (i.e., a symmetry operator that commutes with the Hamiltonian) is introduced. The partition function of the orbifolded theory, $\mathcal{Z}_{\text{orb}}(\tau, \bar{\tau})$, is then defined as:
\begin{equation}
    \mathcal{Z}_{\mathrm{orb}}(\tau,\bar{\tau})=\sum_{g,h\in H,\ [g,h]=e}\tr _{\mathcal{H}_h}\left(g q^{L_0-\frac{c}{24}}\bar{q}^{\bar{L_0}-\frac{c}{24}}\right)
\end{equation}
where the sum is taken over commuting elements of $H$. The space $\mathcal{H}h$ refers to the subspace of the Hilbert space that is invariant under $h$. In this way, terms like $\text{Tr}_{\mathcal{H}_h}(g q^{L_0 - \frac{c}{24}}\bar{q}^{\bar{L_0}-\frac{c}{24}})$, where $g$ is non-trivial, are referred to as the $g$-twisted sectors. Thus, the process of calculating an orbifold boils down to the following algebraic procedure: First, identify the subspace of the Hilbert space fixed by a certain element $h \in H$. In the context of CFT, since the Hilbert space is classified by the highest weight representations of the Virasoro primary fields, this identification can be carried out by determining the subspace of the momentum lattice. The next step is to construct the sector where an element $g$ commuting with $h$ acts on that subspace.
The question of how the parameters of the orbifolded theory are modified is an interesting one, and for the orbifold considered here, this can be derived through the argument presented below.

Let us return to the discussion of the shift symmetry which consists of the action as in equation \eqref{eq:2.4}.
First, let us identify the invariant sector under the $\theta$ shift of this action. Specifically, we need to determine when $\exp(in\theta) \mapsto \exp(in(\theta+\frac{2\pi}{N}))$ remains invariant, that is, for which values of $n$ this holds. Clearly, this occurs when $n$ is a multiple of $N$.
And then, we need to identify the sector in which certain operators belonging to the Hilbert space are acted upon by an element of $\mathbb{Z}_N$.This sector produces an additional phase $\exp(\frac{2\pi i}{N})$ under the periodic shift $\varphi \mapsto \varphi + 2\pi$. This occurs when $w$ is a rational number, specifically when it can be expressed as $w = \frac{w'}{N}$ using an integer $w'$. Here, if we express the general momentum $(p_L, p_R)$ as
\begin{equation}\label{eq:2.7}
    (p_L,p_R)=\left( \frac{n'N}{R}+\frac{w'}{N}R, \frac{n'N}{R} - \frac{w'}{N}R \right)
\end{equation}
this sector can be considered a momentum space with a radius of $\frac{1}{N}R$. By performing a similar analysis for the $\phi$ shift, we will find that the combination of the invariant sector and twisted sector corresponds to a theory with a radius of $\frac{W}{N}R$.

\subsection{The duality defect}
In this section, we will consider the case where the orbifold described in Section \ref{sec:2.1} becomes self-dual and discuss what kind of duality defect appears. First, we need to explain why a duality defect emerges in the self-dual case. This is based on a method called half-space gauging, which follows the logic described below.

First, suppose there is a CFT $\mathcal{T}$ in spacetime. We divide this spacetime into two halves and apply the orbifold operation by a group $G$ on one half. If the orbifolded theory $\mathcal{T} /G$ is equivalent to the original theory, i.e., $\mathcal{T} \simeq \mathcal{T} /G$, then after dividing the spacetime in half and imposing Dirichlet boundary condition, only the theory $\mathcal{T}$ remains in spacetime with the boundary left in place (See Fig.\ref{fig:1}). If the map that identifies $\mathcal{T}$ with $\mathcal{T} /G$ is not the identity map, the boundary state $\ket{\psi}$ on the left and right sides of the boundary will differ, so when an operator on the left side of the boundary moves to the right, it transforms into another operator. This implies that this boundary has a nontrivial action on local operators, and since the map identifying the theories is an isomorphism, it preserves the spectrum. Therefore, this boundary functions as a symmetry of the theory, known as a duality defect (For more details, see \cite{choi_noninvertible_2022}).
\begin{figure}
    \centering
    \includegraphics[width=1.0\linewidth]{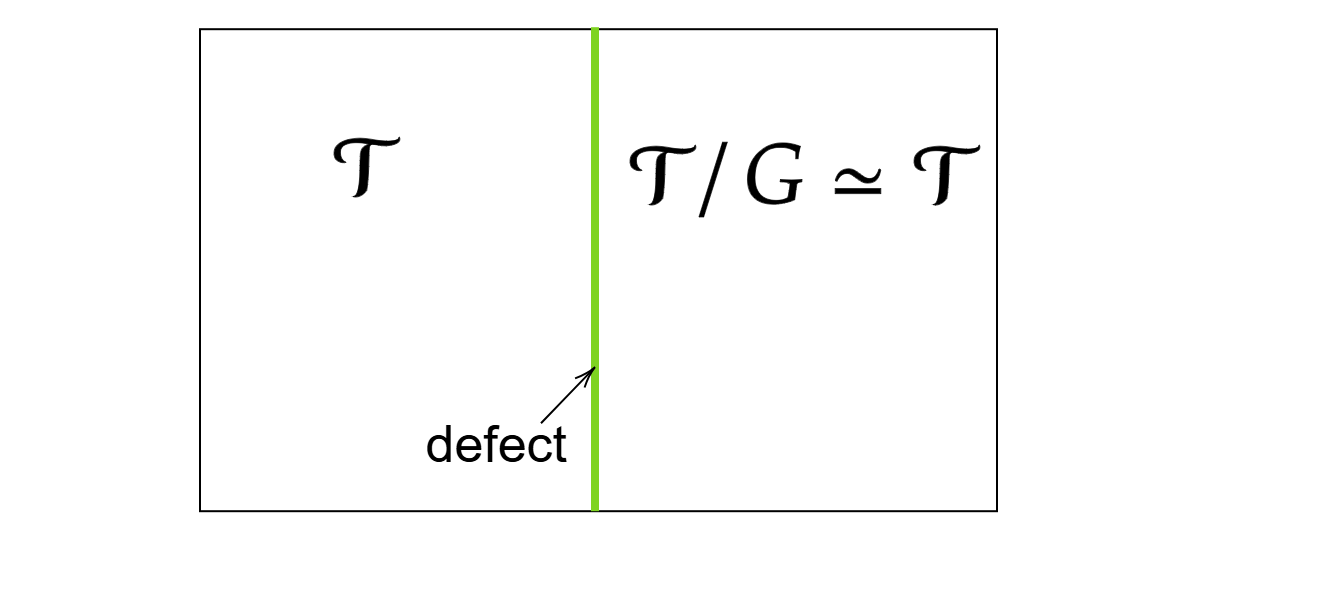}
    \caption{Pictorial realization of half-space gauging.}
    \label{fig:1}
\end{figure}

The above discussion is quite general, but how can the duality defect be calculated in the specific CFT we are focusing on? This question is equivalent to asking what the map identifying the theories is, and it can be concluded that it is T-duality. T-duality is the automorphism group of the lattice formed by the momentum $(p_L, p_R)$ in the CFT. In the case of $c=n\ (n\in\mathbb{Z})$, since this lattice has a metric of signature $(n, n)$, this automorphism group is $\mathrm{O}(n, n; \mathbb{Z})$. Indeed, such transformations change the momentum to different values and can thus be viewed as a transformation from one operator to another. Therefore, the information we need to know is twofold: which orbifold was performed and to which theory did it transform under T-duality, and by combining these, we can calculate the action of the duality defect on operators.

Let us further comment on the specific action of duality defects on local operators. The existence of a defect in the spacetime of a theory implies that the theory must satisfy certain boundary conditions on the defect. Moreover, these boundary conditions must meet specific consistency requirements. For instance, in the case of topological defects, the chiral and anti-chiral parts of the energy-momentum tensor across two regions separated by the defect must satisfy:
\begin{equation}
    T^1(z)-T^2(z)=\bar{T}^1(\bar{z})-\bar{T}^2(\bar{z}),
\end{equation}
indicating that the defect does not transfer energy and can be smoothly deformed. \footnote{If the CFT in question has extended affine Lie algebra symmetry, the boundary state must satisfy conditions preserving the current, as discussed in \cite{cardy_boundary_1989}.}

Such duality defects in CFTs have been extensively studied in the past. For example, \cite{fuchs_tft_2002} utilized duality defects to construct CFT correlators and orbifold theories. Additionally, more recent research on compact boson CFTs includes the work in \cite{argurio_non-invertible_2024,fuchs_topological_2007}.

Now, how exactly do these defects act on local operators? 
Historically, these defects are studied in the context of the folding trick\cite{wong_tunneling_1994,oshikawa_defect_1996}. Such defects, or conformal interfaces, induces a map from the Hilbert space of one side to the other. (See \cite{quella_reflection_2007} for more datails.) To summerize,
the action depends on how the orbifolded theory $\mathcal{T}/G$ is identified with the original theory $\mathcal{T}$—that is, on the ``isomorphism" between $\mathcal{T}/G$ and $\mathcal{T}$. In this context, we consider T-duality as the duality connecting the two theories. Therefore, the task is to determine how operators in the orbifolded theory $\mathcal{T}/G$ correspond to operators in the original, non-orbifolded theory $\mathcal{T}$ under T-duality.
It is worth noting, as will be discussed in later examples, that such a correspondence may include differences in overall scaling factors.

\subsection{Identifying the duality symmetries}\label{sec:2.3}
Now, let us return to the focus point of this section, the case of $c=1$. In this case, T-duality is, apart from a difference in convention regarding constants, a transformation that maps the radius $R$ to $\frac{1}{R}$, with no other possibilities. Therefore, the problem we need to consider is: ``For which values of $R$, $N$, and $M$ does $\frac{M}{N}R$ coincide with $\frac{1}{R}$?"
The simplest example is the following. First, let us consider the case where the radius is $\sqrt{k}$ (with $k \in \mathbb{Z}$). Clearly, when $N=k$ and $M=1$, it is easy to see that the theory becomes self-dual. In other words, the radius of the orbifolded theory is $\frac{1}{k}\sqrt{k}=\frac{1}{\sqrt{k}}=\frac{1}{R}$, meaning that it is T-dual to the original theory. Of course, a similar argument applies when $k$ is a rational number, and a situation where the theory is self-dual can be created by performing an orbifold with $M \neq 1$.

Now that we have confirmed the existence of a new duality defect in this case, we should ask: what kind of fusion rule governs the original symmetry and this defect? The term fusion rule here refers to embedding two symmetry elements as submanifolds, i.e., defects, in spacetime and examining their interaction. In other words, it is the defect obtained in the limit where two defects approach each other, which is an extension of the product structure obtained by the OPE (operator product expansion) of local operators. In this case, let us denote the newly created defect as $\mathcal{D}$ and the defect corresponding to the orbifolded symmetry $\mathbb{Z}_k$ as $\eta_a \ (a \in \mathbb{Z})$. Although the detailed discussion is provided in \cite{choi_self-duality_2024}, let us try to find these fusion rules. First, we will verify the fusion rule between $\eta_a$ and $\mathcal{D}$ for any $a \in \mathbb{Z}$.

In this theory, $\mathbb{Z}_k$ is gauged from the state of a global symmetry to a state considered as a gauge field. Therefore, the defect $\mathcal{D}$ is transparent to the action of $\eta_a$, meaning that
\begin{equation}
    \eta_a \times\mathcal{D}=\mathcal{D}\times\eta_a =\mathcal{D}.
\end{equation}
Also, it holds that
\begin{equation}
    \mathcal{D}\times\mathcal{D}=\sum_{a\in\mathbb{Z}_k}\eta_a .
\end{equation}
This is because the defect does not depend on orientation, so it is equal to its own orientation reversal $\bar{\mathcal{D}}$, meaning that $\mathcal{D} \times \mathcal{D} = \mathcal{D} \times \bar{\mathcal{D}}$. The result of this fusion is equal to taking the limit where the two defects collide, with a mesh of the algebra generated by the symmetry of the orbifold (in this case, the sum of defects generating the $\mathbb{Z}_k$ symmetry) inserted between them. Hence, the above result follows. (For details, see \cite{kaidi_symmetry_2023,thorngren_fusion_2024-2,thorngren_fusion_2024-3})

Let us review how the defects arising in a $\mathbb{Z}_k$-orbifold act on local operators. Here, local operators can be expressed in the form of vertex operators, so we denote them as $V_{n,w}$, where $(n,w)$ represents the charges introduced in eq.\eqref{eq:2.7}. Consider the case where $N = k$ and $M = 1$. How do $(n, w)$ transform under this orbifold?

The $\mathbb{Z}_k$-symmetry acts on $\exp(i n \theta)$ as $\exp(i n (\theta + \frac{2\pi}{k}))$. Therefore, only sectors where $n$ is a multiple of $k$ are relevant in the orbifolded theory. Operators from sectors where $n$ is not a multiple of $k$ map to non-local operators. If $n$ can be written as $n = n' k$ $(n' \in \mathbb{Z})$, then $\exp(i n \theta) = \exp(i n' (k \theta))$, and the momentum charge for this sector is described by $n'$. Thus, $n$ transforms as $n \to n' = \frac{n}{k}$.
Now, consider the winding charge $w$. Since the twisted sector is described by rational numbers with $k$ in the denominator, we have $w = \frac{w'}{k}$ $(w' \in \mathbb{Z})$. Hence, the winding charge $w$ transforms as $w \to w' = k w$.
The duality relating the orbifolded theory $\mathcal{T}/G$ to the original theory $\mathcal{T}$ is T-duality. In the $c = 1$ case, T-duality interchanges $n$ and $w$, as evident from the transformation $R \to \frac{1}{R}$ in eq.\eqref{eq:2.7}. Combining these arguments, the action of the duality defect is given by:
\begin{equation}
    (n,w)\mapsto(w’,n’)=(kw,\frac{n}{k}).
\end{equation}
An important point to note is that, at the operator level, there exists an overall multiplicative constant. This arises due to the non-trivial braiding between $\theta$ and $\phi$. Specifically, since operators include a product structure considering normal ordering, exchanging the positions of two variables corresponds to a $\pi$-rotation of $\theta$ around $\phi$. For simplicity, if $\phi$ is inserted at the origin, this results in a $\pi$-rotation of the argument of $\theta$. By recalling the Baker-Campbell-Hausdorff formula, the overall factor can be shown to be proportional to $(\exp(i \pi))^{nw}$, where the phase originates from expanding the operators in terms of Laurent series and logarithms.

Taking into account this phase and an additional constant arising from the defect's quantum dimension, the complete transformation is:
\begin{equation}\label{eq:2.12}
    V_{n,w}\mapsto \sqrt{k}\exp(i\pi nw)V_{kw,\frac{n}{k}}\ (\mathrm{for} \ n\in k\mathbb{Z}).
\end{equation}
Here, the leading factor $\sqrt{k}$ arises from calculating the quantum dimension of the defect. From here on, we will disregard such differences in constant multipliers in our discussion. Thus, the action of the duality defect can be interpreted as a transformation on the two-component vector $(n, w)^{\top}$ using a matrix. This point will be revisited when discussing the $c = 2$ case\footnote{Though the shift symmetry can be extended to $c\geq 2$ cases, it is not always true that there is a duality defect corresponding to $\mathbb{Z}_k$ shift symmetry.\cite{schweigert_kramerswannier_2008}.}, where the corresponding matrix is given by $\begin{pmatrix} 0 & k \\ \frac{1}{k} & 0 \end{pmatrix}$.

\section{The Main Section: Duality Defects in $c=2$}\label{sec:3}

\subsection{Parameters of the moduli space and T-duality}\label{sec:3.1}
Before diving into the main topic, it is essential to outline the distinctions between the $c=2$ case and the $c=1$ case reviewed in the previous section. To begin, let us consider the CFT itself. Simply put, with two compactified bosons, there are now two components in the momentum, so it takes the form:
\begin{equation}
    \begin{pmatrix}
        p_L \\ p_R
    \end{pmatrix}=\frac{1}{\sqrt{2}}
    \begin{pmatrix}
        \Lambda^{\ast}& (I-B)\\
        \Lambda^{\ast}& (I+B)
    \end{pmatrix}
    \begin{pmatrix}
        n_1 \\ n_2 \\ w_1 \\ w_2
    \end{pmatrix}
\end{equation}
where $\Lambda^{\ast}$ is the inverse of the transpose of $\Lambda$ and $n_1,n_2,w_1,w_2\in\mathbb{Z}$. In theb following, we use the notation $n=\begin{pmatrix}
    n_1\\n_2
\end{pmatrix},\ w=\begin{pmatrix}
    w_1 \\ w_2
\end{pmatrix}$.
Additionally, each of $\theta$ and $\phi$ now also have two components, so there are four $U(1)$ symmetries in this theory. From here on, we denote these symmetries as $U(1)_{n_1} \times U(1)_{n_2} \times U(1)_{w_1} \times U(1)_{w_2}$. The lattice $\Lambda$ is the generating matrix required to specify the two-dimensional torus on which the two bosons are compactified. Additionally, $B$ is a second-rank antisymmetric tensor analogous to the Kalb-Ramond field in type IIB superstring theory. In the $c=1$ case, only a trivial $B$ exists, but when $c$ is greater than or equal to 2, non-trivial $B$ effects can lead to additional complexities. For this reason, for example, \cite{nagoya_non-invertible_2023} considers only the case where $B=0$.

The most critical question is what parameters are necessary to specify the theory, and here, multiple choices exist. Of these, the choice most compatible with our approach involves introducing two complex numbers $(\tau, \rho)$, each with a positive imaginary part, defined as:
\begin{equation}\label{eq:3.2}
    \tau= \frac{G_{12}}{G_{11}}+\frac{\sqrt{\mathrm{det}G}}{G_{11}}
    \quad \rho=b+\sqrt{\det G}.
\end{equation}
In this setup, $G$ represents the torus metric defined by $G = \Lambda^{\top} \Lambda$, where the components $G_{ij}$ for $(i,j=1,2)$ correspond to elements of this Gram matrix for the compactified bosons. Meanwhile, $b$ represents the off-diagonal components of the antisymmetric $2 \times 2$ matrix given by $\Lambda^{\top} B \Lambda$. Here, $B$ serves as an analog of the Kalb-Ramond field in string theory, and its interaction with $\Lambda$ introduces modifications that affect the structure of the lattice and, thus, the duality transformations within the CFT. Introducing these parameters, often referred to as moduli, allows us to express the T-duality of this CFT as follows:
\begin{equation}
    \mathrm{P}(\mathrm{SL}(2,\mathbb{Z})_{\tau}\times\mathrm{SL}(2,\mathbb{Z})_{\rho})\rtimes (\mathbb{Z}_2 ^M \times \mathbb{Z}_2^I)
\end{equation}
Here, $\mathrm{SL}(2, \mathbb{Z})_{\tau}$ acts on $\tau$ according to
\begin{equation}
    \mathrm{SL}(2,\mathbb{Z})_{\tau}\ni
    \begin{pmatrix}
        \alpha & \beta \\
        \gamma & \delta \\
    \end{pmatrix}:\quad
    \tau\mapsto\frac{\alpha\tau +\beta}{\gamma\tau +\delta}
\end{equation}
with a similar action on $\rho$.
Additionally, each of the two $\mathbb{Z}_2$ symmetries acts as follows:
\begin{equation}
    \mathbb{Z}_2 ^M: (\tau,\rho)\mapsto(\rho,\tau),\quad \mathbb{Z}_2^I : (\tau,\rho)\mapsto (-\bar{\tau},-\bar{\rho}).
\end{equation}
For instance, $\tau$ serves as the moduli parameter for the complex structure of the torus compactifying the bosons, while $\rho$ represents the parameter for the Kähler structure. Therefore, $M$ is recognized as one type of Mirror symmetry. Meanwhile, $I$ is a spacetime inversion symmetry, as it is achieved by reversing the torus orientation. Thus, all T-dualities can be viewed as combinations of the two $\mathrm{SL}(2, \mathbb{Z})$ groups, along with four possible combinations of Mirror symmetry and spacetime inversion, ``twisted" in the product by the adjoint action of the two $\mathbb{Z}_2$ symmetries on $\mathrm{SL}(2, \mathbb{Z})$.\footnote{Naturally, the product of these elements cannot be decomposed into a simple product of $\mathrm{SL}(2, \mathbb{Z})^2$ and $\mathbb{Z}^2$, so a straightforward product law does not apply. However, as elements (i.e., when ignoring the semidirect product structure), any element can be expressed using these combinations.} The crucial point in our discussion is that T-dualities can be classified into four types based on the presence or absence of $M$ and $I$. Hence, instead of addressing all duality defects simultaneously, we approach the problem by dividing them into these four types, enabling us to tackle a generally challenging problem. In the following, we denote by $(m,i)\ (m,i=0,1)$ the
value of the two $\mathbb{Z}_2$ components $\mathbb{Z}_2^M\times\mathbb{Z}_2^I$.

The orbifold considered here involves only finite subgroup symmetries that independently rotate each of the four $U(1)$ symmetries, specifically:
\begin{equation}
    \mathbb{Z}_{N_i}\in U(1)_{n_i}, \quad 
\mathbb{Z}_{W_i}\in U(1)_{w_i}, \ i=1,2
\end{equation}
We thus focus on an orbifold generated by the direct product of these four finite abelian groups. Our next step is to clarify the transformation that these symmetries induce on the parameters $\tau$ and $\rho$.
To proceed, let us recall the definitions of $\theta$ and $\phi$:
\begin{equation}
\begin{split}
    \theta&= \frac{\Lambda^{-1}}{\sqrt{2}}(X_L+X_R)\\
    \phi&= \frac{\Lambda^{\top}}{\sqrt{2}}((I-B)X_L-(I+B)X_R)
\end{split}
\end{equation}
To simplify, we first consider the orbifolding by $\mathbb{Z}_{N_1}$. This involves combining the invariant and twisted sectors, where $\theta_1 \mapsto N_1 \theta_1$ and $\phi_1 \mapsto \frac{1}{N_1} \phi_1$, while keeping the other components unchanged. Applying this transformation to the previous expression allows us to determine how $\Lambda$ should transform. Specifically, let us represent the components of the $2 \times 2$ matrix $\Lambda$ as
\begin{equation}
    \Lambda=\begin{pmatrix} \lambda_1 & \lambda_2 \\ \lambda_3 & \lambda_4 \end{pmatrix}
\end{equation}
with the inverse matrix given by
\begin{equation}
    \Lambda ^{-1}=\frac{1}{\det \Lambda}\begin{pmatrix} \lambda_4 & -\lambda_2 \\ -\lambda_3 & \lambda_1 \end{pmatrix}
\end{equation}
Since $\theta_1 \mapsto N_1 \theta_1$, the first row of $\Lambda^{-1}$ must be scaled by $N_1$. Noting the term $\frac{1}{\det \Lambda}$ in $\Lambda^{-1}$, we see that $\lambda_3$ and $\lambda_1$ must each be scaled by $\frac{1}{N_1}$. This corresponds to scaling the first column of $\Lambda$ by $\frac{1}{N_1}$, which implies that one of the basis vectors spanning the lattice $\Lambda$ is scaled by $\frac{1}{N_1}$. Consequently, the Gram matrix $G$ of this basis transforms as follows:
\begin{equation}
    G_{11}\mapsto \frac{1}{N_1^2}G_{11},\ G_{12}\mapsto\frac{1}{N_1}G_{12},\ G_{22}\mapsto G_{22}
\end{equation}
Therefore, by equation (3.2), we obtain
\begin{equation}
    \tau\mapsto N_1\tau
\end{equation}
On the other hand, for the first component of $\phi$, denoted by $\phi_1$, the orbifold operation transforms it as $\phi_1 \mapsto \frac{1}{N_1} \phi_1$. Note that we can rewrite
\begin{equation}\label{eq:3.12}
    \phi = \frac{\Lambda^{\top}}{2}(X_L - X_R) - \frac{\Lambda^{\top}B\Lambda}{2}\theta.
\end{equation}
The first term is consistent with the above discussion, where only $\phi_1$ is scaled by $\frac{1}{N_1}$, and $\phi_2$ remains unchanged. Therefore, the question is which variable should be transformed and how to achieve the effect of scaling only the first component of the second term by $\frac{1}{N_1}$. The answer is that $b$ should be scaled by $\frac{1}{N_1}$. In fact, when
\begin{equation}
    \Lambda^{\top}B\Lambda=
    \begin{pmatrix} 0 & b \\ -b & 0 
    \end{pmatrix} \mapsto 
    \begin{pmatrix} 0 & \frac{b}{N_1} \\ -\frac{b}{N_1} & 0
    \end{pmatrix}
\end{equation}
we see that scaling the first component of $\theta$ by $N_1$ results in the first component of the second term in equation \eqref{eq:3.12} being scaled by $\frac{1}{N_1}$. Therefore, from eq.\eqref{eq:3.2} again, we find
\begin{equation}
    \rho\mapsto\frac{1}{N_1}\rho.
\end{equation}
The same logic applies to the other abelian groups $\mathbb{Z}_{N_2, W_1, W_2}$. For instance, in the case of a $\mathbb{Z}_{N_2}$ orbifold, scaling the second column of $\Lambda$ by $\frac{1}{N_2}$ will yield the following transformation in $G$:
\begin{equation}
    G_{11}\mapsto G_{11},\ G_{12}\mapsto \frac{1}{N_2}G_{12},\ G_{22}\mapsto\frac{1}{N_2^2}G_{22}.
\end{equation}
Similarly, $b$ will be scaled by $\frac{1}{N_2}$. Summing up these transformations, we obtain the following proposition:
\begin{prop}
    Let $\mathcal{T}$ be the theory of the compactified boson CFT at $c=2$, with moduli parameters given by $(\tau, \rho)$. Consider the finite abelian group
    \begin{equation}
        \mathbb{Z}_{N_i}\in U(1)_{n_i}, \quad \mathbb{Z}_{W_i}\in U(1)_{w_i}, \ i=1,2
    \end{equation}
    which induces shifts in $\theta$ and $\phi$, and let $\mathcal{T}'$ be the orbifolded theory by this group. If $(\tau', \rho')$ are the parameters for $\mathcal{T}'$, then
    \begin{equation}
        \tau '=\frac{N_1 W_2}{N_2 W_1}\tau ,\quad \rho '=\frac{W_1 W_2}{N_1 N_2}\rho.
    \end{equation}
\end{prop}

As before, the transformations induced on $(\tau, \rho)$ by the orbifold have become clear, but it remains to determine how these transformations affect the charges $(n, w)$, as was done in Section \ref{sec:2.3}. Additionally, the action of T-duality must be understood, which is significantly more complex compared to the $c=1$ case and thus not easily written explicitly. However, this is discussed in detail in \cite{damia_exploring_2024}, and we will review their arguments.

Let us first consider the effect of the orbifold. To do so, recall how each $\mathbb{Z}_{N_i, W_i}$ orbifold scales $(\theta_1, \theta_2, \phi_1, \phi_2)$. For example, under the subgroup $\mathbb{Z}_{N_1} \times \mathbb{Z}_{W_1} \subset U(1){n_1} \times U(1)_{w_1}$, it follows from the discussion in Section \ref{sec:2.3} that the charges transform into\footnote{Note that $n,w$ are vectors with two elements.}
\begin{equation}
    \begin{pmatrix}
        n' \\ w'
    \end{pmatrix}= \begin{pmatrix}
         \frac{N_1}{W_1} & & & \\ & 1 & & \\
         & & \frac{W_1}{N_1} & \\ & & & 1
    \end{pmatrix}^{-1}
    \begin{pmatrix}
        n \\ w
    \end{pmatrix}.
\end{equation}
Thus, we assign a matrix $\sigma_{N_1, W_1}$ to the orbifold induced by the subgroup $U(1){n_1} \times U(1){w_1}$ and another matrix $\tilde{\sigma}$ to the subgroup $U(1){n_2} \times U(1){w_2}$. The definition is as follows:
\begin{equation}\label{eq:new3.19}
    \sigma_{N_1,W_1}=\begin{pmatrix}
         \frac{N_1}{W_1} & & & \\ & 1 & & \\
         & & \frac{W_1}{N_1} & \\ & & & 1
    \end{pmatrix},
    \tilde{\sigma}_{N_2,W_2}=\begin{pmatrix}
         1 & & & \\ & \frac{N_2}{W_2} & & \\
         & & 1 & \\ & & & \frac{W_2}{N_2}
    \end{pmatrix}
\end{equation}
and the action on the four-component vector $(n', w')$ is given by the standard product operation under $(\sigma^\top)^{-1}$.

Next, consider T-duality. In the $c=1$ case, T-duality was simply represented by the matrix $\begin{pmatrix} 0 & 1 \\ 1 & 0 \end{pmatrix}$. However, for $c=2$, T-duality is described by a combination of two $\SL(2, \mathbb{Z})$ transformations and two $\mathbb{Z}_2$ symmetries, necessitating explicit specification of the corresponding components. This can be done as \cite{damia_exploring_2024}
\begin{equation}\label{eq:new3.20}
\begin{split}
    \begin{pmatrix}
        a & -b & 0 & 0 \\ -c & d & 0 & 0 \\
        0 & 0 & d & c \\ 0 & 0 & b & a
    \end{pmatrix},\ \mathrm{for}\ \begin{pmatrix}
        a & b \\ c & d
    \end{pmatrix}\in \SL(2,\mathbb{Z})\\
    \begin{pmatrix}
        d' & 0 & 0 & c' \\ 0 & d' & -c & 0 \\
        0 & -b' & a' & 0 \\ b' & 0 & 0 & a'
    \end{pmatrix},\ \mathrm{for}\ \begin{pmatrix}
        a' & b' \\ c' & d'
    \end{pmatrix}\in \SL(2,\mathbb{Z})\\
    M=\begin{pmatrix}
        0 & 0 & 1 & 0 \\
        0 & 1 & 0 & 0 \\
        1 & 0 & 0 & 0 \\
        0 & 0 & 0 & 1
    \end{pmatrix}\in \mathbb{Z}_2^M
    ,\quad
    I=\begin{pmatrix}
        -1 & 0 & 0 & 0 \\
        0 & 1 & 0 & 0 \\
        0 & 0 & -1 & 0\\
        0 & 0 & 0 & 1
    \end{pmatrix}\in\mathbb{Z}_2^I
\end{split}
\end{equation}
Using this approach, one can determine the action of duality defects on the charges $(n, w)$, which is equivalent to a linear automorphism on the theory itself. Alternatively, this can be interpreted as a linear transformation on the quadratic form describing the spectrum of the theory.

The specific form of this quadratic form, known as the generalized metric, is described in \cite{becker_string_2006} and similar references. It is obtained by expressing the spectrum of vertex operators $V_{n, w}$ in terms of the charges $(n, w)$. The generalized metric is given by
\begin{equation}\label{eq:3.21}
    \mathcal{G}=\frac{1}{\sqrt{2}}\begin{pmatrix}
        G^{-1} & BG^{-1} \\ -G^{-1}B & G-BG^{-1}B
    \end{pmatrix}  .
\end{equation}
The action of an orbifold or T-duality can then be expressed as $\mathcal{G} \mapsto T^{-1} \mathcal{G} T$, where $T$ is a matrix combining the transformations described by equations \eqref{eq:new3.19} and \eqref{eq:new3.20}.
From this, the condition for realizing a duality defect—namely, that the orbifolded theory is T-dual to itself—can also be written in terms of the generalized metric:
\begin{equation}\label{eq:new3.22}
    \mathcal{G}=T^{-1}\mathcal{G}T.
\end{equation}
Thus, the classification of duality defects can be reduced to solving this matrix equation. However, in practice, this is far from straightforward. In fact, constructing the general solution to such matrix equations is not feasible. From a physical perspective, T-duality swaps $(G - B G^{-1} B) \leftrightarrow G^{-1}$ and $BG^{-1} \leftrightarrow -G^{-1}B$, making it clear that finding transformations inducing such exchanges is not trivial.

That said, the problem simplifies when $B = 0$, allowing general solutions to be constructed using this generalized metric formulation \cite{nagoya_non-invertible_2023}. However, since our focus is on cases with non-vanishing $B$-fields, we must consider alternative methods. These methods, involving certain quadratic equations, are introduced in the following sections.

\subsection{Quadratic equations for duality defects}\label{sec:3.2}
From here, we will address our main question: ``When is the orbifolded theory T-dual to the original theory?" This inquiry is motivated by our desire to classify the types of duality defects present in the theory. In line with this motivation, we will classify the duality defects into four types, as we did in Section \ref{sec:3.1}. Specifically, we will first divide the cases based on the values of $(m,i)$, and finally summarize the results of these classifications.

\subsubsection{Duality defects for $(m,i)=(0,0)$}\label{sec:3.2.1}
let us first consider the simplest case, $(m,i) = (0,0)$, where the T-duality between the orbifolded theory and the original theory only involves $\mathrm{SL}(2,\mathbb{Z})$. Written in terms of an equation, this becomes:
\begin{equation}\label{eq:3.18}
    \frac{N_1 W_2}{N_2 W_1}\tau = \frac{\alpha \tau +\beta}{\gamma \tau +\delta},\quad \frac{W_1 W_2}{N_1 N_2}\rho =
    \frac{\alpha '\rho + \beta '}{\gamma '\rho + \delta '}
\end{equation}
where $\begin{pmatrix}
    \alpha & \beta \\ \gamma & \delta
\end{pmatrix}$ is an elements of $\mathrm{SL}(2,\mathbb{Z})$ and so as $\begin{pmatrix}
    \alpha ' & \beta ' \\ \gamma '& \delta '
\end{pmatrix}$. The problem posed by this equation is: What values for the orders of the orbifold $N_i, W_i$ ensure the existence of $\alpha, \beta, \gamma, \delta, \alpha', \beta',\gamma,'\delta'$ that satisfy the equation \eqref{eq:3.18}?

When such values of $N_i, W_i$ are given, calculating the corresponding element of $\mathrm{SL}(2,\mathbb{Z})$ helps determine the fusion rule for that particular defect. We have simplified the problem by transforming it into a form that makes it much easier to solve.

Before going into the theorem, we introduce the elemantary lemma which is essential to the proof:
\begin{lem}\label{lem:3.1}
    Let $f(x)\in \mathbb{R}[x]$ and $\mathrm{deg}f=2$, and $f(x)=0$ have 
    a solution $z\in\mathbb{C}$ with $\Im z\neq 0$.
    Then the other solution of the quadratic equation $f(x)=0$ is $\bar{z}$, meaning that 
    \begin{equation}
        f(x)=C(x-z)(x-\bar{z}),\ C\in\mathbb{R}.
    \end{equation}
\end{lem}

Then our main theorem is the following:

\begin{thm}\label{thm:3.1}
    The existence of $\alpha, \beta, \gamma, \delta, \alpha', \beta',\gamma,'\delta'$ that satisfy the equation \eqref{eq:3.18} above is equivalent to the following statement:
\begin{screen}
    Two quadratic equations in $x, y, x', y'$:
\begin{equation}\label{eq:3.19}
    \begin{split}
    (N_2 W_1)^2 x^2 -2N_1 N_2 W_1 W_2 (\Re \tau) xy + (N_1 W_2)^2 |\tau|^2 y^2 
    = N_1 N_2 W_1 W_2 \\
    (N_1 N_2)^2 x'^2 -2N_1 N_2 W_1 W_2 (\Re \rho) x'y' + (W_1 W_2)^2 |\rho|^2y'^2
    = N_1 N_2 W_1 W_2
\end{split}
\end{equation}

have integer solutions $(x_0, y_0, x'_0, y'_0)$. 
In addition, given these solutions, the four numbers
\begin{equation}\label{eq:3.20}
\begin{split}
     z_0=-\frac{N_1 W_2}{N_2 W_1}|\tau|^2 y_0,\quad
     w_0=-2(\Re \tau) y_0+ \frac{N_2 W_1}{N_1 W_2}x_0,\\
     z'_0=-\frac{W_1 W_2}{N_1 N_2}|\rho|^2 y'_0,\quad
     w'_0=-2(\Re \rho) y'_0 +\frac{N_1 N_2}{W_1 W_2}x'_0
\end{split}
\end{equation}
are all integers.
\end{screen}

In this case, the corresponding $\mathrm{SL}(2,\mathbb{Z})$ elements are:
\begin{equation}\label{eq:3.26}
    \begin{split}
        \alpha =x_0,\ \beta =z_0,\ \gamma =y_0,\ \delta =w_0,\\
        \alpha'=x'_0,\ \beta'=z_0,\ \gamma'=y'_0,\ \delta'=w'_0
    \end{split}
\end{equation}
which corresponds to the desired statement.
\end{thm}
\begin{proof}
    First, we consider only $\tau$, which is the first part
    of eq. \eqref{eq:3.18}, and deforming the equation, we obtain
    \begin{equation}\label{eq:3.22}
        \begin{split}
            &\frac{N_1 W_2}{N_2 W_1}\tau = \frac{\alpha \tau +\beta}{\gamma \tau +\delta} \\
        \Leftrightarrow & \frac{N_1 W_2}{N_2 W_1}\tau
        \left(\gamma \tau+\delta
        \right)= \alpha\tau +\beta \\
        \Leftrightarrow & 
        \frac{N_1 W_2}{N_2 W_1}\gamma\tau^2 +
        \left(
        \frac{N_1 W_2}{N_2 W_1}\delta -\alpha
        \right)\tau -\beta=0,
        \end{split}
    \end{equation}
which implies that $\tau$ is one of a solution of a quadratic equation (The case $\gamma=0$ when this is no longer quadratic will be discussed later.)
\begin{equation}\label{eq:3.23}
        \frac{N_1 W_2}{N_2 W_1}\gamma x^2 +
        \left(
        \frac{N_1 W_2}{N_2 W_1}\delta -\alpha
        \right)x -\beta=0.
\end{equation}
What is important is that this is a quadratic equation of real valued coefficients, as we can make use of Lemma \ref{lem:3.1}.
Applying Lemma\ref{lem:3.1} to the equation \eqref{eq:3.23},
we find 
\begin{equation}
    \frac{N_1 W_2}{N_2 W_1}\gamma x^2 +
        \left(
        \frac{N_1 W_2}{N_2 W_1}\delta -\alpha
        \right)x -\beta=
        C(x-\tau)(x-\bar{\tau})
\end{equation}
Therefore comparing two coefficients of both side, we have
\begin{equation}\label{eq:new3.26}
    \begin{split}
        &C=\frac{N_1 W_2}{N_2 W_1}\gamma,\ 
        \frac{N_1 W_2}{N_2 W_1}\delta -\alpha=-2C\Re\tau,\ 
        -\beta=C|\tau|^2\\
        \Leftrightarrow & 
        \delta=
        \frac{N_2 W_1}{N_1 W_2}\alpha -2(\Re\tau)\gamma,\ 
        \beta =-|\tau|^2\frac{N_1 W_2}{N_2 W_1}\gamma.
    \end{split}
\end{equation}
In order to guarantee that $\det\begin{pmatrix}
    \alpha & \beta \\ \gamma & \delta
\end{pmatrix}=1$, we must impose $\alpha\delta-\beta\gamma=1$.
Hence, replacing $\delta,\beta$ for ones in eq.\eqref{eq:new3.26},
we have
\begin{equation}
    \begin{split}
        & \alpha\delta- \beta\gamma =1 \\
        \Leftrightarrow & 
        \alpha\left( \frac{N_2 W_1}{N_1 W_2}\alpha -2(\Re\tau)\gamma
        \right)
        +|\tau|^2\frac{N_1 W_2}{N_2 W_1}\gamma =1\\
        \Leftrightarrow & (N_2 W_1)^2\alpha -2(\Re\tau)N_1 N_2 W_1 W_2 \alpha\gamma +(N_1 W_2)^2|\tau|^2\gamma =
        N_1 N_2 W_1 W_2,
    \end{split}
\end{equation}
which means that $\alpha$ (resp. $\gamma$) is a solution of the first line of eq.\eqref{eq:3.19} for $x$ (resp. $y$).
There is another condition we must impose. It is 
that $\alpha,\beta,\gamma,\delta$ are all integers.
To impose this, we should say the solution of eq.\eqref{eq:3.19} must be integers, and additionally impose $\beta,\delta$ are
integer. That is why there is an additional condition as in eq.\eqref{eq:3.20}. 
Let us consider the case $\gamma=0$. From eq.\eqref{eq:3.22}
one finds 
\begin{equation}
    \alpha=\frac{N_1 W_2}{N_2 W_1}\delta,\quad \beta=0
\end{equation}
because $\Im\tau\neq 0$.
Moreover, due to $\alpha\delta=1$, $\alpha,\delta =\pm 1$
and we choose $\alpha=1$ without loss of generality.
Then we get $N_1 W_2=N_2 W_1$ and this is equivalent to 
that the first line of eq.\eqref{eq:3.19} has a solution 
$(x,y)=(1,0)$, which means that the statement of the theorem holds for $\gamma=0$.
To end this proof, we consider the equation for $\rho$:
\begin{equation}
    \begin{split}
        &\frac{W_1 W_2}{N_1 N_2}\rho =\frac{\alpha '\rho+\beta'}{\gamma'\rho +\delta}\\
        \Leftrightarrow & \frac{W_1 W_2}{N_1 N_2}\gamma'\rho^2
        +\left(\frac{W_1 W_2}{N_1 N_2}\delta' -\alpha' 
        \right)-\beta'=0.
    \end{split}
\end{equation}
This case can be proved by the same way as for $\tau$, replacing
$N_1 \leftrightarrow W_1$ and $\tau\leftrightarrow\rho$.
Then the statement of the theorem follows because the opposite direction of the statement can be proved by tracing back this argument.
\end{proof}
Moreover, note that the correspondence between the eight integers and the two $\SL(2, \mathbb{Z})$ elements, as shown in this equation \eqref{eq:3.26}, remains consistent not only for this theorem but also for all subsequent theorems up to Theorem \ref{thm:3.4}. Therefore, this correspondence will be omitted in the discussions that follow.

Using quadratic equations to explore duality defects has a key advantage: we avoid solving complex matrix equations. Generally, solving equations like equation \eqref{eq:new3.22} is challenging. However, by using equations such as equation \eqref{eq:3.19}, we can easily address simpler questions, like when $(x, y) = (0, 1)$ is a solution. (This topic will be discussed in later sections.)

To perform a consistency check, let us verify that performing a trivial orbifold with $N_1 = N_2 = W_1 = W_2 = 1$ results in the theory being identical to itself—that is, yielding the identity element of $\mathrm{SL}(2, \mathbb{Z})$. Substituting $N_1 = N_2 = W_1 = W_2 = 1$ into equation \eqref{eq:3.19} gives
\begin{equation}
    \begin{split}
        x^2-2(\Re\tau)xy+|\tau|^2y^2=1\\
        x'^2-2(\Re\rho)x'y'+|\rho|^2y'^2=1,
    \end{split}
\end{equation}
which clearly has solutions $x = 1, y = 0, x' = 1, y' = 0$ for any values of $\tau$ and $\rho$. For this solution, one can see that the $\SL(2,\mathbb{Z})$ element is the identity matrix.

As an additional consideration, when either the absolute value of $\tau$ or $\rho$ is equal to 1, this equation has solutions with $x = 0, y = 1$ or $x' = 0, y' = 1$. Such solutions suggest the presence of a duality defect involving S-duality as an element of $\mathrm{SL}(2, \mathbb{Z})$. This insight is difficult to obtain through the matrix representation using a generalized metric, underscoring the utility of the quadratic equation approach in analyzing duality defects.

\subsubsection{Duality defects for $(m,i)=(1,0)$}\label{sec:3.2.2}
Next, we discuss the case where the duality includes a non-trivial mirror symmetry, specifically $(m,i) = (1,0)$. When a non-trivial $\mathbb{Z}_2$ action on $(\tau, \rho)$ is present, we cannot directly apply the arguments from Section \ref{sec:3.2.1}. In the previous argument, we could use a quadratic equation satisfied by $\tau$ or $\rho$, as both sides of the equation in eq.\eqref{eq:3.18} involved the same parameter, either $\tau$ or $\rho$. In this case, however, the parameters on both sides differ, preventing us from straightforwardly creating a quadratic equation.

To address this, we can exploit the fact that $\tau$ and $\rho$ are linearly dependent over $\mathbb{R}$. Using equation \eqref{eq:3.2}, there exist real numbers $p$, $q$, and $t$ such that
\begin{equation}
    p\tau =\rho+t.
\end{equation}
In the case where the CFT is rational\footnote{This is equivalent for the compactifying torus to have complex multiplication \cite{gukov_rational_2004}. For the explanation of rationality of compact boson CFT, see \cite{wendland2000moduli,moore_attractors_2003,moore_arithmetic_2007}.}, meaning that there exists a negative integer $D$ such that $\tau, \rho \in \mathbb{Q}(\sqrt{D})$, there exist integers $p$, $q$, and $t$ such that
\begin{equation}\label{eq:3.32}
    p\tau=q\rho +t.
\end{equation}
In general, we only need to introduce $q$ when dealing with rational cases, but for convenience, we will use three real numbers, $p$, $q$, and $t$, in both rational and non-rational cases, setting $q = 1$ in non-rational cases. This relationship between $\tau$ and $\rho$ allows us to align the parameters on both sides as in equation \eqref{eq:3.18}.

Specifically, the orbifold operation results in $(\tau, \rho) \mapsto \left(\frac{N_1 W_2}{N_2 W_1} \tau, \frac{W_1 W_2}{N_1 N_2} \rho \right)$, and applying $\mathrm{SL}(2, \mathbb{Z})$ to $(\tau, \rho)$ followed by mirror symmetry yields $(\tau, \rho) \mapsto \left(\frac{\alpha' \rho + \beta'}{\gamma' \rho + \delta'}, \frac{\alpha \tau + \beta}{\gamma \tau + \delta} \right)$. Combining these, we have
\begin{equation}
        \frac{N_1 W_2}{N_2 W_1}\tau=\frac{\alpha’\rho +\beta’}{\gamma’\rho +\delta'},\quad
        \frac{W_1 W_2}{N_1 N_2}\rho=\frac{\alpha\tau +\beta}{\gamma\tau +\delta},
\end{equation}
and using $p \tau = q \rho + t$, we obtain
\begin{equation}\label{eq:3.34}
    \frac{N_1 W_2}{pN_2 W_1}(q\rho+t)=\frac{\alpha’\rho +\beta’}{\gamma’\rho +\delta'},\quad
    \frac{W_1 W_2}{qN_1 N_2}(p\tau-t)=\frac{\alpha\tau +\beta}{\gamma\tau +\delta}.
\end{equation}
With this setup, we arrive at the following theorem.
\begin{thm}\label{thm:3.2}
    The existence of $\alpha, \beta, \gamma, \delta, \alpha', \beta',\gamma,'\delta'$ that satisfy the equation \eqref{eq:3.34} above is equivalent to the following statement:
    \begin{screen}
        Two quadratic equations in $x, y, x', y'$:
\begin{equation}\label{eq:3.35}
    \begin{split}
    p(N_2 W_1)^2 x^2 -2pN_1 N_2 W_1 W_2 (\Re \tau) xy + p(N_1 W_2)^2 |\tau|^2 y^2 
    = qN_1 N_2 W_1 W_2 \\
    q(N_1 N_2)^2 x'^2 -2qN_1 N_2 W_1 W_2 (\Re \rho) x'y' + q(W_1 W_2)^2 |\rho|^2y'^2
    = pN_1 N_2 W_1 W_2
\end{split}
\end{equation}

have integer solutions $(x_0, y_0, x'_0, y'_0)$. 
In addition, given these solutions, the four numbers
\begin{equation}\label{eq:3.36}
\begin{split}
     z_0=\frac{t}{q}x_0-\frac{pN_1 W_2}{qN_2 W_1}|\tau|^2 y_0,\quad
     w_0=-\left(2(\Re \rho)+\frac{t}{q} \right)y_0+ \frac{pN_2 W_1}{qN_1 W_2}x_0,\\
     z'_0=-\frac{t}{p}x_0'-\frac{qW_1 W_2}{pN_1 N_2}|\rho|^2 y'_0,\quad
     w'_0=-\left(2(\Re \tau)-\frac{t}{p} \right)y'_0 +\frac{qN_1 N_2}{pW_1 W_2}x'_0
\end{split}
\end{equation}
are all integers.
    \end{screen}
\end{thm}
\begin{proof}
    Similarly for Theorem \ref{thm:3.1}, we begin with considering
    $\tau$ and deform eq.\eqref{eq:3.34} to obtain
    \begin{equation}
        \begin{split}
            &\frac{W_1 W_2}{qN_1 N_2}(p\tau -t)=\frac{\alpha\tau +\beta}{\gamma\tau +\delta} \\
            \Leftrightarrow &
            \frac{W_1 W_2}{qN_1 N_2}(p\tau -t)(\gamma\tau +\delta)
            -\alpha\tau -\beta =0\\
            \Leftrightarrow &
            \frac{pW_1 W_2}{qN_1 N_2}\gamma \tau^2 +\left(
            \frac{W_1 W_2}{qN_1 N_2}(p\delta-t\gamma)-\alpha
            \right)\tau-\left(\frac{W_1 W_2}{qN_1 N_2}t\delta+\beta\right)=0
        \end{split}
    \end{equation}
using Lemma \ref{lem:3.1} again, two equations
\begin{equation}\label{eq:3.38}
    -2\frac{pW_1 W_2}{qN_1 N_2}(\Re\tau)\gamma=\frac{W_1 W_2}{qN_1 N_2}(p\delta-t\gamma)-\alpha,\quad
    -\frac{pW_1 W_2}{qN_1 N_2}|\tau|^2\gamma=\frac{W_1 W_2}{qN_1 N_2}t\delta+\beta
\end{equation}
holds. Then, from the former part of eq.\eqref{eq:3.38}, 
\begin{equation}
    \delta=
    \frac{qN_1 N_2}{pW_1 W_2}\alpha -\left(2(\Re\tau)-\frac{t}{p}\right)\gamma.
\end{equation}
Therefore we have
\begin{equation}
    \begin{split}
    \beta &=-\frac{t}{p}\alpha-\frac{W_1 W_2}{pqN_1 N_2}
    (p^2|\tau|^2-2(\Re\tau)pt+t^2)\gamma\\
    &= -\frac{t}{p}\alpha-\frac{W_1 W_2}{pqN_1 N_2}((p\Re\tau-t)^2+(p\Im\tau)^2)\gamma\\
    &= -\frac{t}{p}\alpha-\frac{qW_1 W_2}{pN_1 N_2}|\rho|^2\gamma
    \end{split}
\end{equation}
where we used $p\Re\tau-t=q\Re\rho,\ p\Im\tau=q\Im\tau$ due to eq.\eqref{eq:3.32}.
Combining this result with $\alpha\delta-\beta\gamma=1$ yields
\begin{equation}
    \begin{split}
        &\alpha\left( \frac{qN_1 N_2}{pW_1 W_2}\alpha -\left(2(\Re\tau)-\frac{t}{p}\right)\gamma
        \right)
        +\left(\frac{t}{p}\alpha+\frac{qW_1 W_2}{pN_1 N_2}|\rho|^2\gamma\right)\gamma=1\\
        \Leftrightarrow &
        \frac{qN_1 N_2}{pW_1 W_2}\alpha^2 -\frac{2}{p}(\Re p\tau-t)\alpha\gamma+\frac{qW_1 W_2}{pN_1 N_2}|\rho|^2\gamma^2=1\\
        \Leftrightarrow &
        (N_1 N_2)^2\alpha^2-2N_1 N_2 W_1 W_2(\Re\rho)\alpha\gamma
        +(W_1 W_2)^2|\rho|^2\gamma^2=\frac{p}{q}N_1 N_2 W_1 W_2,
    \end{split}
\end{equation}
which is equivalent to the first line of eq.\eqref{eq:3.35}.
The remaining discussion is completely analog of the proof of Theorem \ref{thm:3.1}.
Hence, we see the statement follows.
\end{proof}

\subsubsection{Duality defects for $(m,i)=(0,1)$}\label{sec:3.2.3}
Next we go on to determining two quadratic equations for the case $(m,i)=(0,1)$ where there is
a nontrivial spacetime inversion $(\tau,\rho)\mapsto (-\bar{\tau},-\bar{\rho})$.
In this case, we can also use the same logic as \ref{sec:3.2.2}, in which we relate
the parameters $(\tau,\rho)$ to each other by scalar multiplication and shift.
If we write $\tau=\tau_1 + i\tau_2,\ \rho=\rho_1+\rho_2$ with $\tau_i,\rho_i \in\mathbb{R}\ (i=1,2)$, 
$-\bar{\tau}$ and $-\bar{\rho}$ are represented as
\begin{equation}
    -\bar{\tau}=\tau-2\tau_1,\quad -\bar{\rho}=\rho-2\rho_1.
\end{equation}
Therefore two parameters with $\mathbb{Z}$ acted can be related to the original ones,
and setting the coefficient in eq.\eqref{eq:3.34} as
$p=q=1, t=-2\tau_1$ for the argument of $\tau$, $p=q=1, t=-2\rho$ for $\rho_1$,
one can proceed with the discussion of this case.

\begin{thm}\label{thm:3.3}
    The existence of $\alpha, \beta, \gamma, \delta, \alpha', \beta',\gamma,'\delta'$ that satisfy the equation \eqref{eq:3.34} above is equivalent to the following statement:
    \begin{screen}
        Two quadratic equations in $x, y, x', y'$:
\begin{equation}\label{eq:3.43}
    \begin{split}
    (N_2 W_1)^2 x^2 -2N_1 N_2 W_1 W_2 (\Re \tau) xy + (N_1 W_2)^2 |\tau|^2 y^2 
    = N_1 N_2 W_1 W_2 \\
    (N_1 N_2)^2 x'^2 -2N_1 N_2 W_1 W_2 (\Re \rho) x'y' + (W_1 W_2)^2 |\rho|^2y'^2
    = N_1 N_2 W_1 W_2
\end{split}
\end{equation}

have integer solutions $(x_0, y_0, x'_0, y'_0)$. 
In addition, given these solutions, the four numbers
\begin{equation}\label{eq:3.44}
\begin{split}
     z_0=-\frac{2\Re \rho}{q}x_0-\frac{pN_1 W_2}{qN_2 W_1}|\tau|^2 y_0,\quad
     w_0=-\left(2(\Re \rho)-\frac{2\Re\rho}{q} \right)y_0+ \frac{N_2 W_1}{N_1 W_2}x_0,\\
     z'_0=\frac{2\Re\tau}{p}x_0'-\frac{qW_1 W_2}{pN_1 N_2}|\rho|^2 y'_0,\quad
     w'_0=-\left(2(\Re \tau)+\frac{2\Re\tau}{p} \right)y'_0 +\frac{N_1 N_2}{W_1 W_2}x'_0
\end{split}
\end{equation}
are all integers.
    \end{screen}
\end{thm}
A notable point is that the quadratic equation \eqref{eq:3.43} that must be satisfied is identical to that in Theorem \ref{thm:3.1}. As shown in Theorem \ref{thm:3.2}, the equations we need to solve can generally be reduced to the same type of equation, which greatly simplifies the calculations and enhances clarity. However, it is crucial to note that additional conditions vary in each case. Specifically, the resulting matrix $\begin{pmatrix}\alpha & \beta \\ \gamma & \delta \end{pmatrix}$ must satisfy conditions to be an element of $\mathrm{SL}(2, \mathbb{Z})$, and these conditions differ depending on each situation.

\subsubsection{Duality defects for $(m,i)=(1,1)$}
Finally, let us consider the fourth case, namely $(m, i) = (1, 1)$. As before, we can obtain an equation for classifying the duality defect. In this case, since the operation involves a combination of mirror symmetry and spacetime inversion, we need to combine the discussions in Sections \ref{sec:3.2.2} and \ref{sec:3.2.3}. Due to this combination, $(\tau, \rho)$ will map to $(-\bar{\rho}, -\bar{\tau})$, so we aim to express these parameters in terms of the original parameters $(\tau, \rho)$. We start by choosing $(p, q, t) \in \mathbb{R}$ such that
\begin{equation}
    p\tau=q\rho+t.
\end{equation}
Consequently, we find that
\begin{equation}
    -\bar{\tau}=\frac{q\rho-t-2q\Re\rho}{p},\quad -\bar{\rho}=\frac{p\tau+t-2p\Re\tau}{q}.
\end{equation}
Therefore, equation (3.18) modifies to
\begin{equation}\label{eq:3.47}
    \frac{N_1 W_2}{N_2 W_1}\frac{p\tau+t-2p\Re\tau}{q}=\frac{\alpha\tau +\beta}{\gamma\tau+\delta},\quad
\frac{W_1 W_2}{N_1 N_2}\frac{q\rho-t-2q\Re\rho}{p}=\frac{\alpha’\rho +\beta'}{\gamma’\rho +\delta'}.
\end{equation}
This shifts $t$ to $t - 2p\Re \tau$ for $\tau$ and to $t + 2q\Re \rho$ for $\rho$, mirroring the discussion in Section \ref{sec:3.2.2} for both $\tau$ and $\rho$.

The rest of the argument proceeds in the same way as in the previous sections, allowing us to obtain the final theorem of this sequence of theorems.
\begin{thm}\label{thm:3.4}
    The existence of $\alpha, \beta, \gamma, \delta, \alpha', \beta',\gamma,'\delta'$ that satisfy the equation \eqref{eq:3.34} above is equivalent to the following statement:
    \begin{screen}
        Two quadratic equations in $x, y, x', y'$:
\begin{equation}\label{eq:3.48}
    \begin{split}
    p(N_2 W_1)^2 x^2 -2pN_1 N_2 W_1 W_2 (\Re \tau) xy + p(N_1 W_2)^2 |\tau|^2 y^2 
    = qN_1 N_2 W_1 W_2 \\
    q(N_1 N_2)^2 x'^2 -2qN_1 N_2 W_1 W_2 (\Re \rho) x'y' + q(W_1 W_2)^2 |\rho|^2y'^2
    = pN_1 N_2 W_1 W_2
\end{split}
\end{equation}

have integer solutions $(x_0, y_0, x'_0, y'_0)$. 
In addition, given these solutions, the four numbers
\begin{equation}\label{eq:3.49}
\begin{split}
     z_0=\frac{t}{q}x_0-\frac{pN_1 W_2}{qN_2 W_1}|\tau|^2 y_0,\quad
     w_0=-\left(2(\Re \rho)+\frac{t + 2q\Re \rho}{q} \right)y_0+ \frac{N_2 W_1}{N_1 W_2}x_0,\\
     z'_0=-\frac{t}{p}x_0'-\frac{qW_1 W_2}{pN_1 N_2}|\rho|^2 y'_0,\quad
     w'_0=-\left(2(\Re \tau)-\frac{t - 2p\Re \tau}{p} \right)y'_0 +\frac{N_1 N_2}{W_1 W_2}x'_0
\end{split}
\end{equation}
are all integers.
    \end{screen}
\end{thm}
With the above considerations, we have now established a method to determine all duality defects generated by the orbifold of a diagonal discrete subgroup $\mathbb{Z}_{N_1} \times \mathbb{Z}_{N_2} \times \mathbb{Z}_{W_1} \times \mathbb{Z}_{W_2}$ of $U(1)^4$. Specifically, by substituting a given $(\tau, \rho)$ into the set of four quadratic equations, we can find values of $(N_1, N_2, W_1, W_2)$ that admit solutions under the imposed conditions. In the following, we will verify this procedure with simple examples.

\subsection{Some brief examples and restrictions for existence of solutions}\label{sec:3.3}

In Section \ref{sec:3.2}, we derived a set of four quadratic equations, but two points make solving the problem more challenging. First, these are quadratic equations in two variables, and it is generally impossible to explicitly write down their integer solutions. Furthermore, we impose an additional condition that other real numbers, expressed in terms of these solutions, must also be integers (see Theorems \ref{thm:3.1} to \ref{thm:3.4}). Finding a general solution that satisfies all these conditions is not straightforward, making the classification of duality defects a difficult problem. However, as shown in the next section, it is relatively easier to classify defects belonging to a specific class. Before discussing that, we will illustrate how to generate defects using a simple example.

\paragraph{$(\tau,\rho)=(i,i)$}
Let us consider the simplest case, $(\tau, \rho) = (i, i)$, which corresponds to a theory composed of two compactified $c=1$ CFTs with radius 1. Additionally, the theory with $R=1$ has a chiral algebra $\mathrm{SU}(2)$, and by overlapping the left and right sectors, it attains an overall symmetry of $\mathrm{SU}(2) \times \mathrm{SU}(2)$. This theory is thus known as a symmetry-enhanced point with $\mathrm{SU}(2)^2 \times \mathrm{SU}(2)^2$ symmetry. Although the symmetry of this theory is classified in the conventional sense, let us use our method to explore categorical symmetries that extend it, specifically those that are non-invertible.

First, let us look for defects that emerge in the simplest case of $(m, i) = (0, 0)$. Noting that $\Re \tau = \Re \rho = 0$ and $|\tau| = |\rho| = 1$, the first equation in \eqref{eq:3.19} becomes
\begin{equation}
    (N_2 W_1)^2 x^2+(N_1 W_2)^2 y^2=N_1 N_2 W_1 W_2.
\end{equation}
Dividing both sides by $N_1 N_2 W_1 W_2$, and for simplicity denoting positive coprime integers $K_1, K_2$ by $\frac{N_1 W_2}{N_2 W_1} = \frac{K_1}{K_2}$, we obtain
\begin{equation}\label{eq:3.51}
    \frac{K_2}{K_1}x^2+\frac{K_1}{K_2}y^2=1.
\end{equation}
According to equation \eqref{eq:3.20}, the solution $(x_0, y_0)$ of this equation must satisfy
\begin{equation}
    \frac{K_2}{K_1}x_0 \in\mathbb{Z},\quad\frac{K_1}{K_2}y_0 \in\mathbb{Z}.
\end{equation}
Since $K_1$ and $K_2$ are coprime, we must express $x_0 = d_1 K_1$ and $y_0 = d_2 K_2$ using integers $d_1$ and $d_2$. Substituting these back into equation \eqref{eq:3.51}, we find
\begin{equation}
    K_1 K_2(d_1^2+d_2^2)=1
\end{equation}
This implies $K_1 = K_2 = 1$ and $(d_1, d_2) = (0, 1)$ or $(1, 0)$. Hence, $N_1 W_2 = N_2 W_1$, and the solutions are $(x, y) = (1, 0)$ or $(0, 1)$. Note that $(x, y) = (1, 0)$ corresponds to the identity matrix in $\mathrm{SL}(2, \mathbb{Z})$, while $(0, 1)$ corresponds to the S-duality matrix.

The same logic applies to the second equation of \eqref{eq:3.19}, leading to $N_1 N_2 = W_1 W_2$, which also has the identity matrix or S-duality as solutions. Considering $\mathrm{GCD}(N_i, W_i) = 1$, the only possible combination is $(N_1, N_2, W_1, W_2) = (1, 1, 1, 1)$. Thus, the duality defects at this point for $(m, i) = (0, 0)$ consist of four types formed by combining S-duality acting on $\tau$ and S-duality acting on $\rho$ \footnote{For convenience, we also count defects that act as identity on local operators}.

The same approach can be applied when considering mirror symmetry and spacetime inversion. Since $\tau = \rho$, the theory is invariant under $\mathbb{Z}_2^M$ symmetry, and since $-\bar{i} = i$, it is also invariant under spacetime inversion. Choosing $(p, q, t)$ such that $p \tau = q \rho + t$, we find $p = q = 1$ and $t = 0$, leading back to the same equation \eqref{eq:3.19}. Therefore, by considering all cases, we generate a total of 16 duality defects as a combination of four S-dualities with or without mirror and spacetime inversion.


As one becomes familiar with such examples, it becomes apparent that the typical solutions to the equations we derived often reduce to $(x, y) = (1, 0)$ or $(0, 1)$. Determining the conditions for $N_1, N_2, W_1, W_2$ under which these solutions occur is relatively straightforward. This is because, when either $x$ or $y$ is zero, only one term remains on the left-hand side of our quadratic equation. From this, it becomes easy to calculate the relationships that $N_1, N_2, W_1, W_2$ must satisfy. Furthermore, how T-duality can be constructed using the S-matrix and T-matrix is also straightforward to compute. This is because, when $x$ or $y$ is zero, the corresponding $\SL(2, \mathbb{Z})$ element is represented by a matrix of the form $\begin{pmatrix} 1 & \ast \\ 0 & 1 \end{pmatrix}$ or $\begin{pmatrix} 0 & -1 \\ 1 & \ast \end{pmatrix}$, making it either a pure power of the T-matrix or a single application of the S-matrix multiplied by a power of the T-matrix. This proves very useful for calculating fusion rules, though this paper does not pursue such computations.

Given the above, it becomes a mathematically intriguing question whether for any choice of $N_1, N_2, W_1, W_2$, there always exists a $(\tau, \rho)$ such that one of the solutions $(x, y)$ necessarily has $x = 0$ or $y = 0$. Addressing this question is also crucial for the physical problem of symmetry classification. Particularly, for the case $(m, i) = (0, 0)$, we provide an answer as follows\footnote{While this mentions only $\tau$, this also holds for $\rho$.} (noting that the essence of the problem does not change for other cases):
\begin{prop}\label{prop:3.2}
    Let $\tau$ be a complex number satisfying $\Re \tau \geq 0$, $(\Im \tau)^2 > \Re \tau + \frac{1}{4}$ or $\Re \tau < 0$, $(\Im \tau)^2 > -\Re \tau + \frac{1}{4}$. Then, for any positive integers $N_1, N_2, W_1, W_2$, if the quadratic equation
    \begin{equation}\label{eq:3.59}
        \frac{N_1 W_2}{N_2 W_1}x^2 -2(\Re\tau)xy+\frac{N_2 W_1}{N_1 W_2}|\tau|^2y^2=1
    \end{equation}
    has an integer solution for $(x,y)$, then it must satisfy $x = 0$ or $y = 0$.
\end{prop}
The proof is deferred to the appendix \ref{appendix:b}. We make a few remarks about this proposition. First, while this proposition addresses the equation in terms of $\tau$, an identical statement holds for equations involving $\rho$. That is, for any choice of orbifold parameters, the sufficient conditions on $\rho$ for the existence of a solution where $x = 0$ or $y = 0$ are the same as those for $\tau$. Moreover, although this proposition is stated for the quadratic equation corresponding to $(m, i) = (0, 0)$, similar propositions can be proved for the quadratic equations in other cases. In such cases, the range of $\tau$ and $\rho$ satisfying the conditions will depend on the choice of $p$ and $q$, but this dependency can be easily calculated by examining the proof.

The proposition’s assumption—that $\Re\tau \geq 0$, $(\Im\tau)^2 > \Re\tau + \frac{1}{4}$ or $\Re\tau < 0$, $(\Im\tau)^2 > -\Re\tau + \frac{1}{4}$—covers a significant portion of the fundamental domain of $\tau$. The fundamental domain is the space of $\tau$ values modulo the action of $\SL(2, \mathbb{Z})$. When this domain is visualized (see Figure \ref{fig:3.1}), it is evident that these conditions encompass the entirety of the fundamental domain except for the isolated point $\tau=\omega$, where $\omega = e^{\frac{2\pi i}{3}}$. This shows that the proposition imposes strong constraints on the existence of duality defects for nearly all theories.

\begin{figure}
    \centering
    \includegraphics[width=1.0\linewidth]{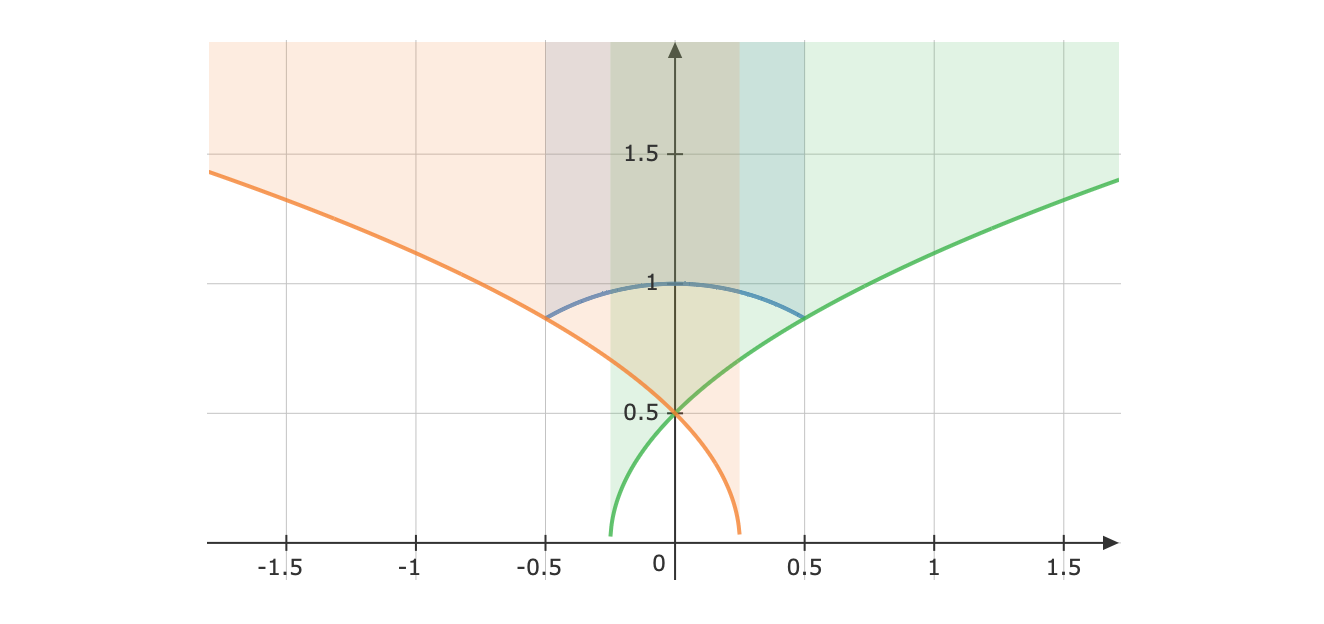}
    \caption{The intersection of the fundamental domein of $\tau$ (blue area) and the assumption of Proposition \ref{prop:3.2} (the intersection of green and orange area). The $x$ axis represent the real part of $\tau$, while $y$ axis does the imaginary part. We can see it covers whole area of the fundamental domain except for $\tau=e^{\frac{2\pi i}{3}}$.}
    \label{fig:3.1}
\end{figure}

As a result, the $(m, i) = (0, 0)$ type duality defects considered in this paper can be completely classified for most examples, as will be demonstrated in the following sections. Furthermore, even in the presence of non-trivial mirror symmetries, if $p = q = 1$, the resulting equation (as shown in equation \eqref{eq:3.59}) remains identical to the one studied here. Thus, in such cases, a complete classification is also achievable.

When $p$ and $q$ take non-trivial values, the coverage extends to most of the fundamental domain. Consequently, strong constraints on the symmetries of many theories can still be established. This illustrates the general applicability of the approach developed in this paper to classify and constrain the existence of duality defects in various settings. We will see such a case in the following.

\paragraph{$(\tau,\rho)=(e^{\frac{2\pi i}{3}},-\frac{1}{2}+\frac{3\sqrt{3}}{2}i)$}
Let us consider the example $(\tau, \rho) = (\omega, \alpha)$, where $\omega = e^{\frac{2\pi i}{3}}$ and $\alpha = -\frac{1}{2} + \frac{3\sqrt{3}}{2}i$, as the final case for this section. This example represents a bicritical point studied in \cite{damia_exploring_2024,dulat_crystallographic_2000}. For instance, the parameter $\omega$ at this point is one of the exceptional cases in Proposition \ref{prop:3.2}, indicating that the quadratic equation \eqref{eq:3.18} admits integer solutions beyond $(x, y) = (1, 0)$ and $(0, 1)$.

What additional solutions exist at such exceptional points? It turns out that only one additional solution appears. The essence of the proof in the appendix \ref{appendix:b} relies on using the arithmetic-geometric mean inequality to show that for an integer solution $(x_0, y_0)$, $x_0 y_0$ is strictly less than 1. The modification in logic for this exceptional point arises because it lies on the boundary of Fig. \ref{fig:3.1}, where all inequalities include equality. This allows $x_0 y_0$ to equal 1, which introduces the solution $x_0 = y_0 = 1$. Consequently, the complete set of solutions is given by $(1, 0)$, $(0, 1)$, and $(1, 1)$.
From this, it is straightforward to see that for the duality symmetry corresponding to $(m, i) = (0, 0)$, the solutions are:
\begin{equation}
    (x,y)=\begin{cases}(1,0)\\ (0,1)\\(1,1) \end{cases},\quad (x’,y’)=\begin{cases} (1,0)\\(0,1)\end{cases}.
\end{equation}
These solutions allow the explicit construction of the duality defect.

The difficult part arises when constructing a defect with mirror symmetry. In this case, $p \tau = q \rho + t$ must hold, where $p = 3$, $q = 1$, and $t = -1$. This means $p = q = 1$ is not satisfied. However, Theorem \ref{thm:3.2} can still be applied here. For an integer solution $(x_0, y_0)$, the upper bound for $x_0 y_0$ changes to a value other than 1. By modifying the proof in the appendix \ref{appendix:b} to use either $\frac{q}{p}$ or $\frac{p}{q}$, depending on the equation considered, the same logic remains valid. Specifically, from the first equation in \eqref{eq:3.35}, we obtain $x_0 y_0 \leq \frac{q}{p}$, and from the second equation, $x'_0 y'_0 \leq \frac{p}{q}$.
Thus, the solutions are:
\begin{equation}
    (x,y)=\begin{cases}(1,0)\\ (0,1) \end{cases},\quad (x’,y’)=\begin{cases} (1,0)\\(0,1)\\(1,1)\\(2,1)\\(1,2)\end{cases}
\end{equation}
By substituting these solutions, we can determine $N_1, N_2, W_1, W_2$.

\section{Symmetries on multicritical points (or lines)}\label{sec:4}
In this section, we expand the results from Section \ref{sec:3.3} by considering solutions not just at a specific point but within a broader ``family" of parameters. However, rather than arbitrary families, we focus on multicritical points and multicritical lines, which are special subsets of the conformal manifold.

The motivation for examining these examples lies in studying the categorical symmetries present on different branches of the moduli space. For instance, let us revisit the bicritical point in the $c=1$ case, also known as the Kosterlitz-Thouless point. At this point, two branches of the moduli space intersect: the circle branch—the family of conventional Narain CFTs—and the orbifold branch—the family of Narain CFTs orbifolded by the charge conjugation $\mathbb{Z}_2^C$. In other words, two possible directions of deformation exist at this point (noting that each branch is one-dimensional), and each direction corresponds to an exactly marginal operator. The deformation along the circle branch is driven by the operator $\partial \phi \bar{\partial}\bar{\phi}$, while deformation along the orbifold branch is associated with $\exp(i\phi) + \exp(-i\phi) \sim \cos(\phi)$ \cite{ginsparg_curiosities_1988}. 
Recall that the KT point exhibits a symmetry described by the Tambara-Yamagami category $\mathrm{TY}(\mathbb{Z}_4)$. This symmetry persists along the orbifold branch because the $\mathrm{TY}(\mathbb{Z}_4)$ symmetry remains intact under deformation in that direction. This conclusion can be reached as follows: if the marginal operator responsible for the deformation is invariant under the action of the symmetry, then that symmetry will be preserved under the deformation. Specifically, in the deformed theory, the marginal operator appears as an insertion in spacetime for computing correlation functions. If this operator commutes with the defect associated with the symmetry, the defect retains its topological nature, preserving the symmetry. Recalling the action \eqref{eq:2.12} of the defects corresponding to $\mathrm{TY}(\mathbb{Z}_4)$, it follows that $\exp(i\phi) + \exp(-i\phi) \sim \cos(\phi)$ is invariant under this action.
Thus, we find that the orbifold branch in the $c=1$ case retains the $\mathrm{TY}(\mathbb{Z}_4)$ symmetry.

What about the $c=2$ case? While we have not yet examined exactly marginal deformations in this context, as a preparatory step, we investigate the duality symmetries present at several multicritical points. This exploration sets the stage for understanding the symmetries preserved under such deformations in $c=2$.
The moduli space for $c=2$ theories has not yet been fully elucidated, but the distribution of multicritical points has been systematically classified in \cite{dulat_crystallographic_2000}. This classification leverages the space group symmetries of crystals to categorize the symmetries present on the toroidal branch (an extension of the $c=1$ circle branch to $c=2$) and computes their orbifolds.

In particular, if a theory obtained by orbifolding a point on the toroidal branch coincides with another theory derived from orbifolding a different point on the same branch, the resulting theory lies at the intersection of two distinct orbifold branches. Such intersections also occur within the toroidal branch itself. Specifically, these are points that can be equivalently described as orbifolds of multiple points on the toroidal branch. A prominent example of this is $(\tau, \rho) = (\omega,\alpha)$, discussed in section \ref{sec:3.3}.

Notably, these intersection points are not always isolated. In some cases, they form a continuous family of points, existing along a line in the moduli space.

\paragraph{$(\tau,\rho)=(it,\frac{1}{2}+it)$ with $t\in \mathbb{R}_{>0}$}
This point is the intersection of two conformal manifolds. One belongs to the toroidal branch, and the other is the branch obtained through a $\mathbb{Z}_2$ orbifold, arising from the lattice shift symmetry of the compactified boson and the sign-reversal symmetry of the boson. Specifically, this theory can be obtained by performing a $\mathbb{Z}_2$ orbifold on the point $(\tau, \rho) = (it, it)$\cite{dulat_crystallographic_2000}.

For instance, substituting these parameters into the quadratic equation \eqref{eq:3.18} yields:
\begin{equation}\label{eq:4.1}
    (N_2 W_1)^2x^2+ t^2(N_1 W_2)y^2=N_1 N_2 W_1 W_2
\end{equation}
Let us first consider the first equation. According to Proposition \ref{prop:3.2}, if this equation has integer solutions, then either $x$ or $y$ must be 0. First, consider the case where $(x, y) = (1, 0)$. Substituting this solution into the equation gives:
\begin{equation}\label{eq:4.2}
    N_2 W_1=N_1 W_2
\end{equation}
In this case, the corresponding element of $\SL(2, \mathbb{Z})$ is the identity matrix. Next, for the case where $(x, y) = (0, 1)$:
\begin{equation}\label{eq:4.3}
    t^2 N_1 W_2 = N_2 W_1
\end{equation}
This implies that $t^2$ must be a rational number, meaning this point corresponds to a rational CFT. The $\SL(2, \mathbb{Z})$ matrix in this case is the S-matrix.

Next, let us focus on the second equation:
\begin{equation}\label{eq:4.4}
    (N_1 N_2)^2x^2-(N_1 N_2 W_1 W_2) xy+ \frac{1+4t^2}{4}(W_1 W_2)y^2=N_1 N_2 W_1 W_2
\end{equation}
If $t \neq \frac{\sqrt{3}}{2}$, the assumptions of Proposition \ref{prop:3.2} are satisfied. Thus, if a solution exists, either $x$ or $y$ must be 0. For $(x, y) = (1, 0)$, the values of $N_1, N_2, W_1, W_2$ satisfy:
\begin{equation}\label{eq:4.5}
    N_1 N_2 = W_1 W_2
\end{equation}
For $(x, y) = (0, 1)$, they satisfy:
\begin{equation}\label{eq:4.6}
    \frac{1+4t^2}{4}W_1 W_2 = N_1 N_2
\end{equation}
The next task is to choose either equation \eqref{eq:4.2} or \eqref{eq:4.3} and either equation \eqref{eq:4.5} or \eqref{eq:4.6}, solve the resulting pair of simultaneous equations under the condition that $N_i, W_i (i=1,2)$ are coprime, and find integer solutions for $N_1, N_2, W_1, W_2$.
For example, let us take the solution $(x,y)=(0,1)$ of equation \eqref{eq:4.1} which corresponds to S-duality on $\tau$ and $(x,y)=(1,0)$ of equation \eqref{eq:4.4} which is the identity on $\rho$.
This means we consider the conditions for $N_1, N_2, W_1, W_2$ such that 
\begin{equation}
    t^2 N_1 W_2 = N_2 W_1\ \mathrm{and}\ N_1 N_2 = W_1 W_2,
\end{equation}
leading us to 
\begin{equation}
    \frac{W_1}{N_1}=\frac{N_2}{W_2}=t.
\end{equation}
Therefore, this condition requires that $t$ to be a rational number, and if this is satisfied, the integer numbers $N_1, N_2, W_1, W_2$ are uniquely determined. Then we get a duality defect composing this orbifold and S-transformation on $\tau$. To summerie this case, we can obtain four duality defect corresponding to each choice of two integer solutions of two equations \eqref{eq:4.1} and \eqref{eq:4.4}. In addition, we can completely classify all duality defects utilizing theorems and propositions in the previous section since $p,q,t$ that satisfy $p\tau=q\rho+t$ for this case is $p=q=1,\ t=-\frac{1}{2}$, which mean we can also use Proposition \ref{prop:3.2}.

\section{Discussion}
Thus, in Theorems \ref{thm:3.1} through \ref{thm:3.4}, we have identified conditions equivalent to self-duality under the orbifold by diagonal discrete subgroups of $U(1)^4$, enabling a more systematic exploration of duality defects. Moreover, these conditions reduce to solving two quadratic equations for integer solutions, from which the specific information about the duality defect—namely, which T-duality defect generates it—can be easily computed. Our additional result, Proposition \ref{prop:3.2}, imposes constraints on the solutions to these quadratic equations, thereby ensuring a complete enumeration of solutions. Consequently, it has become theoretically possible to classify the duality defects (under the orbifolds we have considered) present in a given theory. This represents a decisive step toward the full classification of duality symmetries, including non-invertible ones, in a given theory.

The future work stemming from this study includes the following.
First, the results of this work should be generalized to orbifolds by non-diagonal discrete subgroups. Generators of such discrete subgroups act on vertex primary operators $V_{n,w}$ through integers $k, a_1, a_2, b_1, b_2$ such that:
\begin{equation}
    V_{n,w}\mapsto \exp(\frac{2\pi i}{k}(a_1 n_1+a_2 n_2+b_1 w_1 +b_2 w_2))V_{n,w}.
\end{equation}
This corresponds to the matrices $\sigma$ and $\tilde{\sigma}$ in Equation \eqref{eq:new3.19} being non-diagonal. For such an action, how do $(\tau, \rho)$ transform?

This question could potentially be resolved by examining the action of $\sigma$ and $\tilde{\sigma}$ on $\theta$ and $\phi$. When these matrices are diagonal, as analyzed in Section \ref{sec:3.1}, the transformations correspond to scalar multiples of specific rows in $\Lambda$, enabling the computation of transformations on $\tau$ and $\rho$ via changes in $G$. Determining the transformations of $\Lambda$, i.e., the linear transformations of the lattice's generating vectors, would clarify the transformations of $(\tau, \rho)$. This could then rewrite the self-duality condition for the general case.
If these transformations are linear in $\tau$ and $\rho$, the same logic used to derive Theorems \ref{thm:3.1} through \ref{thm:3.4} would apply, allowing for the classification of duality defects. Resolving this would enable a complete classification of duality defects generated by orbifolds of discrete subgroups of $U(1)^4$.

Beyond shift symmetries, it would be intriguing to consider orbifolds involving symmetries such as the charge conjugation $\mathbb{Z}_2$ symmetry. While we currently lack concrete leads in this area, classifying duality defects arising from such orbifolds would greatly expand the scope of symmetry classifications.

Another potential direction is investigating the duality symmetries preserved during deformations from the toroidal branch to the orbifold branch. As mentioned in Section \ref{sec:4}, this involves explicitly computing the duality defects that preserve the deformation operator. For this computation, it is necessary to determine how the charges $(n, w)$ are acted upon. Using the quadratic equation method we proposed, once solutions are identified, the action can be calculated automatically. The remaining task is to identify the operators that induce exactly marginal deformations toward each orbifold branch.
This, however, is not straightforward in general cases. While an operator being marginal means its conformal dimensions are $(h, \bar{h}) = (1, 1)$, being exactly marginal requires proving that the deformation of the action $\mathcal{S}$ by the operator $\mathcal{O}$:
\begin{equation}
    \delta\mathcal{S}=\lambda\int\dd z\dd\bar{z}\mathcal{O}(z,\bar{z})
\end{equation}
preserves conformal symmetry to all orders in $\lambda$. Demonstrating this is nontrivial and requires analytic calculations. Thus, the key task in this approach is determining when an operator is exactly marginal.
If this is achieved, it would enable the classification of symmetries for a wide range of CFTs, including both the $c=2$ orbifold branch and the toroidal branch.

\appendix

\section{Proof of Proposition \ref{prop:3.2}}\label{appendix:b}
First let us consider the case $\Re\tau = 0$ where we can reduce
eq. \eqref{eq:3.59} to the form
\begin{equation}\label{eq:b.1}
    \frac{N_1 W_2}{N_2 W_1}x^2 +\frac{N_2 W_1}{N_1 W_2}|\tau|^2y^2=1.
\end{equation}
Suppose that this equation has a solution $(x_0,y_0)$ with $x_0 y_0\neq 0$.
Then we can also assume that a solution of this equation satisfies
$x_0>0$ and $y_0>0$ without loss of generality, because for any solution
$(x_0,y_0)$ of eq.\eqref{eq:b.1} all of $(\pm x_0,\pm y_0)$ satisfies it.
Then using the inequality for positive real numbers $a,b$ that states
\begin{equation}
    a+b\geq 2\sqrt{ab}
\end{equation}
where the equality will be attained if and only if $a=b$,
the term $xy$ can be bounded from above as:
\begin{equation}
    \begin{split}
        1&=\frac{N_1 W_2}{N_2 W_1}x_0^2 +\frac{N_2 W_1}{N_1 W_2}|\tau|^2y_0^2\\
        &\geq 2\sqrt{\frac{N_1 W_2}{N_2 W_1}x_0^2 \frac{N_2 W_1}{N_1 W_2}|\tau|^2y_0^2}\\
        &= 2|\tau|x_0y_0\\
        \Leftrightarrow & \frac{1}{2|\tau|} 
        \geq x_0 y_0.
    \end{split}
\end{equation}
Therefore, if $|\tau|>\frac{1}{2}$, one finds
\begin{equation}
    x_0y_0\leq \frac{1}{2|\tau|}<1,
\end{equation}
which implies that $x_0 y_0$ should be zero as $x_0 y_0$ is a nonnegative integer, and this condition matches with the assumption of the proposition $(\Im\tau)^2>\frac{1}{4}$. Then this contradicts $x_0 y_0\neq 0$.

Then suppose that $\tau >0$ and $x_0y_0\neq 0$. Then for this case we can also 
assume that $x_0y_0>0$ because otherwise we have $x_0y_0<0$ and it holds that $-2(\Re\tau)x_0y_0\geq 1$ and then
\begin{equation}\label{eq:b.5}
    1=\frac{N_1 W_2}{N_2 W_1}x_0^2 -2(\Re\tau)x_0y_0+\frac{N_2 W_1}{N_1 W_2}|\tau|^2y_0^2
    \geq 1+1+1>1.
\end{equation}
Without loss of generality we assume $x_0>0$ and $y_0>0$, hence
one finds
\begin{equation}\label{eq:b.6}
    \begin{split}
        1&=\frac{N_1 W_2}{N_2 W_1}x_0^2 -2(\Re\tau)x_0y_0+\frac{N_2 W_1}{N_1 W_2}|\tau|^2y_0^2\\
        &\geq 2\sqrt{\frac{N_1 W_2}{N_2 W_1}x_0^2 \frac{N_2 W_1}{N_1 W_2}|\tau|^2y_0^2}-2(\Re\tau)x_0y_0\\
        &=2(|\tau|-\Re\tau)x_0y_0\\
        \Leftrightarrow & x_0y_0 \leq \frac{1}{2(|\tau|-\Re\tau)}.
    \end{split}
\end{equation}
From the assumption of this proposition, 
\begin{equation}\label{eq:b.7}
    \begin{split}
        (\Im\tau)^2&>\Re\tau +\frac{1}{4}\\
        \Leftrightarrow(\Im\tau)^2 +(\Re\tau^2)&>(\Re\tau)^2 +\Re\tau +\frac{1}{4}\\
        \Leftrightarrow|\tau|^2 &>\left(\Re\tau +\frac{1}{2} \right)^2 \\
       \Leftrightarrow |\tau| -\Re\tau &>\frac{1}{2}
    \end{split}
\end{equation}
where we obtained the last inequality by positivity of both $|\tau|$ and $\Re\tau +\frac{1}{2}$.
Then the two inequalities \eqref{eq:b.6} and \eqref{eq:b.7}
leads to 
\begin{equation}
    x_0 y_0\leq\frac{1}{2(|\tau|-\Re\tau)}<1,
\end{equation}
which contradicts from the assumption. 

Then we close the proof by considering the case $\Re\tau<0$.
However, this case is equivalent to the csae $\Re\tau>0$.
This is because we can assume $x_0 y_0<0$ with the same discussion with eq.\eqref{eq:b.5}, 
and supposing that $x_0<0$ without loss of generality, 
\begin{equation}
\begin{split}
    &\frac{N_1 W_2}{N_2 W_1}x_0^2 -2(\Re\tau)x_0y_0+\frac{N_2 W_1}{N_1 W_2}|\tau|^2y_0^2\\
    =& \frac{N_1 W_2}{N_2 W_1}(-x_0)^2 -2(\Re -\tau)(-x_0)y_0+\frac{N_2 W_1}{N_1 W_2}|-\tau|^2y_0^2.
\end{split}
\end{equation}
Then the problem reduces to the case for $\Re\tau>0$, which
completes the proof.

\acknowledgments
We thank Toshiya Kawai, Yoshiki Fukusumi and Shunta Takahashi for giving comments on this paper. We also thank Masahito Yamazaki, Yuya Kusuki, and Soichiro Shimamori for discussions.
Y.F. is supported by JSPS KAKENHI grant No.23KJ1183.




\bibliographystyle{JHEP}
\bibliography{Furuta2411}






\end{document}